%% file: paper.tex
\documentclass[11pt]{article}

\input{preamble}

\title{Dimension-Free Polylogarithmic Quantum Shadow Tomography}
\author{Fernando Granha Jeronimo\thanks{Department of Computer Science, University of Illinois at Urbana-Champaign. Email: \texttt{granha@illinois.edu}}
\and Qizhao Huang\thanks{Department of Computer Science, University of Illinois at Urbana-Champaign. Email: \texttt{qizhaoh2@illinois.edu}}
\and Lenny Liu\thanks{Department of Computer Science, University of Illinois at Urbana-Champaign. Email: \texttt{hengyu2@illinois.edu}}
}
\date{}

\begin{document}
\maketitle
\begin{abstract}
\textit{Shadow tomography} is a fundamental problem in quantum information theory. Given multiple copies of an unknown $d$-dimensional quantum state $\rho$ and a known collection of observables $E_1,\ldots,E_M$, the goal is to estimate all expectation values $\{\Tr(\rho E_i)\}_{i=1}^M$ to additive accuracy $\varepsilon$ with probability at least $1-\delta$. 

An elusive open question from the seminal shadow tomography work of Aaronson~\cite{aaronson2018shadow} is whether this task admits a dimension-independent sample complexity with only polylogarithmic dependence on $M$, 
as suggested by the best-known lower bounds. In this work, 
we propose two different quantum protocols for shadow tomography with the best sample complexity
\[
O\left(
\frac{\log(M)\log(M/\delta)}{\varepsilon^2}
\right),
\]
which is polylogarithmic in the number of observables and independent of the dimension of the unknown state, thereby answering Aaronson's original question while also providing an exponential improvement in the prior best dimension independent sample complexity of shadow tomography from Sinha~\cite{sinha2024} (STOC 2025) and, more recently, Chen, O'Donnell, Pelecanos, and Wright~\cite{chenOdonnellPelecanosWright2026}.

Our approach first reduces the general shadow tomography problem to a finite-ensemble estimation problem via a minimax argument. We then develop an observable-independent protocol that repeatedly applies the pretty-good measurement while updating the prior distribution over the finite ensemble according to the measurement outcomes. A tail analysis of the resulting estimation error yields simultaneous accuracy guarantees for all observables and a cubic-logarithmic upper bound. We also introduce a refined recovery-label measurement for the same finite ensemble, which yields the bound in our main theorem.
\end{abstract}




\newpage

\input{sections/introduction}

\input{sections/algorithms}
\input{sections/posterior_localization}
\input{sections/recovery_analysis}
\input{sections/worst_case}

\section{Acknowledgement}
\textbf{Use of generative AI.} The authors used OpenAI’s ChatGPT 5.6 Sol and ChatGPT 6 Astra as assistance tools during the preparation of this manuscript. These tools provided suggestions and insights on the use of sequential PGM to improve the analysis, the use of the recovery measurements from~\cite{junge2018recovery}, and the geometric precision refinement developed in Section~\ref{subsec:geometric-refinement}. They were also used to assist with proof-clarity checks, internal notation-consistency checks, and editorial revisions.

The minimax framework and the mathematical analysis of the sequential PGM and recovery measurements are developed by the authors. All mathematical claims, proofs, citations, and final text were reviewed, revised, and verified by the authors, who take full responsibility for the integrity, accuracy, originality, and copyright compliance of the manuscript.
\newpage

\bibliographystyle{alpha}
\begingroup
\small
\bibliography{references}
\endgroup

\appendix
\input{sections/binomial_distribution}
\input{sections/recovery_appendix}

\end{document}

%% file: preamble.tex
\usepackage[T1]{fontenc}
\usepackage{lmodern}
\usepackage[margin=1in]{geometry}
\usepackage{microtype}
\usepackage{amsmath,amssymb,amsthm,mathtools}
\usepackage[numbers,sort&compress]{natbib}
\usepackage{xcolor}
\usepackage{enumitem}
\usepackage{booktabs}
\usepackage{hyperref}
\usepackage{setspace}
\usepackage[capitalise,noabbrev]{cleveref}

\hypersetup{
  colorlinks=true,
  linkcolor=blue!55!black,
  citecolor=green!40!black,
  urlcolor=blue!60!black
}

\allowdisplaybreaks
\setlist{leftmargin=*,topsep=3pt,itemsep=2pt}

\newtheorem{theorem}{Theorem}[section]
\newtheorem{proposition}[theorem]{Proposition}
\newtheorem{lemma}[theorem]{Lemma}
\newtheorem{corollary}[theorem]{Corollary}
\theoremstyle{definition}
\newtheorem{definition}[theorem]{Definition}
\newtheorem{algorithm}[theorem]{Algorithm}
\newtheorem{remark}[theorem]{Remark}

\DeclareMathOperator{\Tr}{Tr}
\DeclareMathOperator{\supp}{supp}

\newcommand{\cD}{\mathcal{D}}
\newcommand{\cE}{\mathcal{E}}
\newcommand{\cH}{\mathcal{H}}
\newcommand{\ketbra}[2]{\lvert #1\rangle\!\langle #2\rvert}
\newcommand{\abs}[1]{\lvert #1\rvert}

\newcommand{\dd}{\,\mathrm{d}}
\usepackage{comment} 
\usepackage{braket}
\usepackage{algorithm}
\usepackage{algpseudocode}

\usepackage{lmodern}
\usepackage{amsmath}

\definecolor{lcInk}{HTML}{202939}
\definecolor{lcMuted}{HTML}{60718A}
\definecolor{lcGate}{HTML}{758298}
\definecolor{lcBlue}{HTML}{2563C7}
\definecolor{lcBlueFill}{HTML}{DCE9FD}
\definecolor{lcOrange}{HTML}{E67424}
\definecolor{lcOrangeFill}{HTML}{FCE6D2}
\definecolor{lcPurple}{HTML}{8548BF}
\definecolor{lcPurpleFill}{HTML}{D8C0EE}
\definecolor{gateblue}{RGB}{23,44,81}

\definecolor{BellPurple}{RGB}{133,72,191}
\definecolor{BellPurpleFill}{RGB}{232,217,248}

\definecolor{GqGreen}{RGB}{22,138,99}
\definecolor{GqGreenFill}{RGB}{217,242,232}

\usepackage{tikz}
\usetikzlibrary{
  arrows.meta, bending, calc, positioning, fit, backgrounds,
  decorations.pathreplacing, decorations.markings,
  decorations.pathmorphing, shapes.geometric, shapes.misc,
  shapes.arrows, shapes.symbols, shadings, fadings, patterns,
  intersections, matrix, quotes, 3d
}
\usepackage{pgfplots}
\pgfplotsset{compat=1.18}
\usepgfplotslibrary{fillbetween,groupplots,colormaps}

\definecolor{stInk}{HTML}{23313F}
\definecolor{stGray}{HTML}{768292}
\definecolor{stMist}{HTML}{EEF2F6}
\definecolor{stBlue}{HTML}{2F6DB5}
\definecolor{stSky}{HTML}{8DBBE8}
\definecolor{stOrange}{HTML}{E08A1E}
\definecolor{stGold}{HTML}{EDBB4A}
\definecolor{stGreen}{HTML}{22977A}
\definecolor{stRed}{HTML}{C8404A}
\definecolor{stPurple}{HTML}{7652A8}
\definecolor{stTeal}{HTML}{1F8A9E}

\tikzset{
  lbl/.style      = {font=\small, text=stInk},
  note/.style     = {font=\footnotesize, text=stGray},
  tinynote/.style = {font=\scriptsize, text=stGray},
  box/.style = {draw=stInk!55, fill=white, rounded corners=3pt,
                line width=0.6pt, inner xsep=5pt, inner ysep=4pt,
                align=center, font=\small, text=stInk},
  tintbox/.style = {box, draw=#1!75!black, fill=#1!9},
  tintbox/.default = stBlue,
  softbox/.style = {box, draw=#1!45, fill=#1!5, line width=0.5pt},
  softbox/.default = stGray,
  headbox/.style = {box, draw=#1!80!black, fill=#1!80!black,
                    text=white, font=\small\bfseries},
  >={Stealth[length=2.2mm,width=1.6mm]},
  arr/.style      = {->, line width=0.7pt, draw=stInk!75},
  thickarr/.style = {-{Stealth[length=3mm,width=2.4mm]},
                     line width=1.3pt, draw=#1},
  thickarr/.default = stInk!75,
  dasharr/.style  = {->, line width=0.6pt, draw=stInk!60,
                     dash pattern=on 2.5pt off 1.5pt},
  implies/.style  = {-{Implies}, double equal sign distance,
                     line width=0.6pt, draw=stInk!80},
  qwire/.style    = {line width=0.6pt, draw=stInk!80},
  cwire/.style    = {double distance=1.1pt, line width=0.45pt,
                     draw=stInk!80},
  gate/.style     = {draw=stInk!80, fill=white, line width=0.6pt,
                     rounded corners=1.5pt, minimum height=6mm,
                     minimum width=7mm, font=\small, inner sep=2pt},
  pt/.style       = {circle, fill=#1, draw=white, line width=0.4pt,
                     inner sep=0pt, minimum size=4.2pt},
  pt/.default     = stInk,
  bigpt/.style    = {circle, fill=#1, draw=white, line width=0.5pt,
                     inner sep=0pt, minimum size=5.6pt},
  bigpt/.default  = stInk,
  starpt/.style   = {star, star points=5, star point ratio=2.3,
                     fill=#1, draw=white, line width=0.35pt,
                     inner sep=0pt, minimum size=9pt},
  starpt/.default = stGreen,
  brace/.style    = {decorate, decoration={brace, amplitude=4pt,
                     raise=1pt}, line width=0.55pt, draw=stInk!70},
  mbrace/.style   = {decorate, decoration={brace, mirror, amplitude=4pt,
                     raise=1pt}, line width=0.55pt, draw=stInk!70},
}

\newcommand{\stPOVMpath}{plot[smooth cycle, tension=0.75] coordinates
  {(0,1) (1.05,0.78) (1.5,0.02) (1.08,-0.82) (-0.02,-1.02)
   (-1.12,-0.72) (-1.5,0.06) (-1.02,0.86)}}
\tikzset{
  pics/meter/.style = {code={
      \draw[line width=0.5pt, draw=stInk]
        (-0.2,-0.07) arc[start angle=160, end angle=20, radius=0.213];
      \draw[line width=0.5pt, draw=stInk, -{Stealth[length=0.9mm,width=0.8mm]}]
        (0,-0.11) -- (0.12,0.12);
  }},
  pics/povmset/.style = {code={
      \path[shade, inner color=stBlue!5, outer color=stBlue!24] \stPOVMpath;
      \draw[draw=stBlue!85!black, line width=0.8pt] \stPOVMpath;
  }},
  pics/statespace/.style = {code={
      \shade[ball color=stMist!60!stGray!35, opacity=0.9] (0,0) circle (1);
      \draw[stInk!55, line width=0.7pt] (0,0) circle (1);
      \draw[stInk!30, line width=0.4pt, dash pattern=on 1.5pt off 1.2pt]
        (1,0) arc[start angle=0, end angle=-180, x radius=1, y radius=0.3];
      \draw[stInk!18, line width=0.4pt, dash pattern=on 1.5pt off 1.2pt]
        (1,0) arc[start angle=0, end angle=180, x radius=1, y radius=0.3];
  }},
}

\pgfplotsset{
  stplot/.style = {
    axis lines=left,
    axis line style={stInk!70, line width=0.6pt,
                     -{Stealth[length=1.8mm,width=1.4mm]}},
    tick style={stInk!70, line width=0.5pt},
    tick label style={font=\footnotesize, text=stInk},
    label style={font=\small, text=stInk},
    every axis plot/.append style={line width=1pt},
    legend style={font=\footnotesize, draw=none, fill=none},
    legend cell align=left,
    clip=false,
  },
}

%% file: sections/introduction.tex
\section{Introduction}
\label{sec:introduction}

Shadow tomography, introduced by Aaronson~\cite{aaronson2018shadow}, asks
how many copies of an unknown quantum state are required to predict the
expectation values of many known measurements.  Let $\cH$ be a
finite dimensional Hilbert space, let $\rho\in\cD(\cH)$ be unknown, and let
    $0 \preceq E_1,\ldots,E_M \preceq I$
be a known list of quantum observables.  The goal is to estimate all quantities
$\operatorname{Tr}(E_1\rho),\ldots,\operatorname{Tr}(E_M\rho)$
to additive accuracy $\varepsilon$, using as few independent copies of
$\rho$ as possible.

\begin{definition}[Shadow tomography]
\label{def:shadow-tomography}
A $T$-copy shadow-tomography strategy for $E_1,\ldots,E_M$ consists of a
POVM $\{N_z\}_{z\in\mathcal Z}$ on $\cH^{\otimes T}$ and a decoder
\[
    g\colon\mathcal Z\to[0,1]^M.
\]
For an input state $\rho$, write
\[
    \mathbb P_\rho[Z=z]
    \coloneqq
    \operatorname{Tr}(N_z\rho^{\otimes T}).
\]
The strategy has accuracy $\varepsilon$ and failure probability $\delta$ if
\[
    \sup_{\rho\in\cD(\cH)}
    \mathbb P_\rho
    \left[
        \max_{1\leq j\leq M}
        \left|
            g_j(Z)-\operatorname{Tr}(E_j\rho)
        \right|
        >
        \varepsilon
    \right]
    \leq
    \delta.
\]
\end{definition}

For every fixed observable list, one measurement and one decoder must work
simultaneously for every input state.  The strategy may be collective across
the $T$ copies, and it may depend on the known list
$E_1,\ldots,E_M$, but it cannot depend on the unknown state $\rho$.  This is
the original list-dependent problem of Aaronson~\cite{aaronson2018shadow}.
It is distinct from the classical-shadow setting, in which one first produces
a reusable classical description using a measurement chosen independently of
the observables that may later be queried~\cite{huangkuengpreskill2020}.

\subsection{Our result}

Our strategy is information-theoretic, we claim that there exists a collective POVM construction related to input observables, while not claiming it can be computationally efficient.

\begin{theorem}[Dimension-free polylogarithmic shadow tomography]
  \label{thm:main}
  There is a universal constant $C<\infty$ such that, for every
  finite-dimensional $\cH$, every $M\geq1$ and observables
  $E_1,\ldots,E_M$ with $0\preceq E_j\preceq I$, and all
  $0<\varepsilon\leq1$ and $0<\delta<1$, there is a shadow-tomography
  strategy using
  \[
  T=O\left(\varepsilon^{-2}\log(M)\log(M/\delta)\right)
  \]
  copies.  The bound is independent of $\dim\cH$.
\end{theorem}
We improved the current best known 
dimension-independent sample complexity of shadow tomography, before which, the best known was $O(\sqrt{M}/\varepsilon^2)$~\cite{chenOdonnellPelecanosWright2026}.

\input{sections/related_work}

\subsection{Proof roadmap and organization}
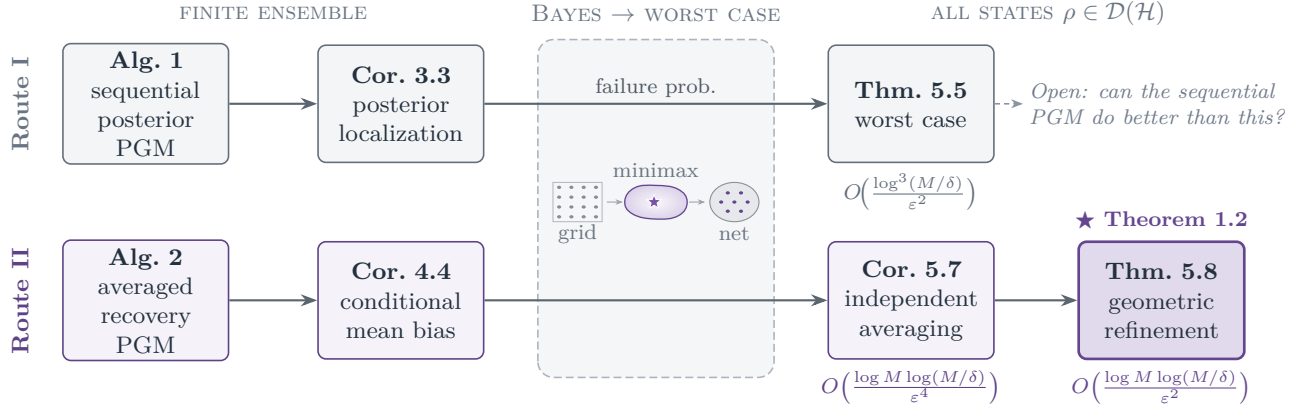
\begin{figure}[H]
    \centering
\begin{tikzpicture}[
  x=1.25cm,
  y=1cm,
  rbox/.style={box, text width=1.98cm, minimum height=1.55cm,
               font=\footnotesize, inner xsep=3pt, inner ysep=3pt},
  r1/.style={rbox, draw=auxC!85!black, fill=auxC!8},
  r2/.style={rbox, draw=mainC!80!black, fill=mainC!7},
  hdr/.style={font=\footnotesize\scshape, text=auxC!90!black},
  flow/.style={arr, line width=0.8pt},
  plab/.style={font=\scriptsize, text=auxC!90!black},
]
\colorlet{mainC}{stPurple}   
\colorlet{auxC}{stGray}      
\def\ya{2.6}   
\def\yb{0}     

\fill[auxC!7, rounded corners=6pt] (5.3,-1.02) rectangle (7.8,3.45);
\draw[auxC!60, line width=0.6pt, rounded corners=6pt, dash pattern=on 3pt off 1.5pt]
     (5.3,-1.02) rectangle (7.8,3.45);
\begin{scope}[shift={(5.72,1.3)}]              
  \foreach \i in {0,...,3}{\foreach \j in {0,...,3}{
     \fill[stInk!55] (\i*0.13-0.195,\j*0.13-0.195) circle (0.022);}}
  \draw[stInk!35, line width=0.3pt] (-0.26,-0.26) rectangle (0.26,0.26);
\end{scope}
\draw[stInk!45, line width=0.4pt, -{Stealth[length=1mm]}] (6.03,1.3) -- (6.2,1.3);
\begin{scope}[shift={(6.55,1.3)}, scale=0.22]  
  \path[shade, inner color=mainC!5, outer color=mainC!26] \stPOVMpath;
  \draw[mainC!85!black, line width=0.5pt] \stPOVMpath;
\end{scope}
\node[starpt=mainC, minimum size=5.5pt] at (6.53,1.30) {};
\draw[stInk!45, line width=0.4pt, -{Stealth[length=1mm]}] (6.9,1.3) -- (7.07,1.3);
\begin{scope}[shift={(7.38,1.3)}]              
  \fill[auxC!18] (0,0) circle (0.26);
  \draw[stInk!55, line width=0.4pt] (0,0) circle (0.26);
  \foreach \a in {0,60,...,300}{\fill[mainC!85!black] (\a:0.14) circle (0.022);}
  \fill[mainC!85!black] (0,0) circle (0.022);
\end{scope}
\node[plab] at (5.72,0.88) {grid};
\node[plab] at (6.55,1.72) {minimax};
\node[plab] at (7.38,0.88) {net};

\node[hdr] at (2.5,3.8) {finite ensemble};
\node[hdr] at (6.55,3.8) {Bayes $\to$ worst case};
\node[hdr] at (10.7,3.8) {all states $\rho\in\mathcal D(\mathcal H)$};

\node[r1] (A1) at (1.15,\ya) {\textbf{Alg.~\ref*{alg:sequential-pgm}}\\ sequential\\ posterior PGM};
\node[r1] (A2) at (3.85,\ya) {\textbf{Cor.~\ref*{cor:simultaneous-finite-prior-cubic}}\\ posterior\\ localization};
\node[r1] (A3) at (9.25,\ya) {\textbf{Thm.~\ref*{thm:sequential-uniform}}\\ worst case};
\draw[flow] (A1) -- (A2);
\draw[flow] (A2) -- (A3) node[pos=0.5, above=1pt, font=\scriptsize, text=stInk!80,
      fill=auxC!7, inner sep=1pt] {failure prob.};

\node[font=\scriptsize, text=auxC!85!black, anchor=north] at (A3.south)
     {$O\!\bigl(\tfrac{\log^3(M/\delta)}{\varepsilon^2}\bigr)$};

\draw[auxC!85, line width=0.6pt, dash pattern=on 2pt off 1.5pt, -{Stealth[length=1.6mm]}]
     (A3.east) -- ++(0.3,0)
     node[right, font=\scriptsize\itshape, text=auxC!95!black, align=left, inner sep=2pt]
     {Open: can the sequential\\ PGM do better than this?};

\node[r2] (B1) at (1.15,\yb) {\textbf{Alg.~\ref*{alg:recovery-pgm}}\\ averaged\\ recovery PGM};
\node[r2] (B2) at (3.85,\yb) {\textbf{Cor.~\ref*{cor:conditional-bias}}\\ conditional\\ mean bias};
\node[r2] (B4) at (9.25,\yb) {\textbf{Cor.~\ref*{cor:independent-averaging}}\\ independent\\ averaging};
\node[r2, draw=mainC!85!black, fill=mainC!22, line width=1pt] (B5) at (11.9,\yb)
      {\textbf{Thm.~\ref*{thm:geometric-refinement}}\\ geometric\\ refinement};

\draw[flow] (B1) -- (B2);
\draw[flow] (B2) -- (B4);
\draw[flow] (B4) -- (B5);

\node[font=\scriptsize, text=mainC!80!black, anchor=north] at (B4.south)
     {$O\bigl(\tfrac{\log M\log(M/\delta)}{\varepsilon^{4}}\bigr)$};
\node[font=\scriptsize, text=mainC!85!black, anchor=north] at (B5.south)
     {$O\bigl(\tfrac{\log M\log(M/\delta)}{\varepsilon^{2}}\bigr)$};
\node[font=\scriptsize\bfseries, text=mainC!85!black, anchor=south] at (B5.north)
     {$\bigstar$ Theorem~\ref*{thm:main}};

\node[font=\footnotesize\bfseries, text=auxC!90!black, rotate=90] at (-0.18,\ya) {Route I};
\node[font=\footnotesize\bfseries, text=mainC!80!black, rotate=90] at (-0.18,\yb) {Route II};
\end{tikzpicture}
\caption{Proof Roadmap}
\end{figure}

Our proof consists of two stages. We first develop two measurement procedures for a finite ensemble of candidate states: a sequential pretty-good measurement (PGM) and a recovery measurement. We then use minimax and continuity arguments to convert their finite-prior guarantees into worst-case shadow-tomography guarantees that hold uniformly over all input states. The recovery-measurement route ultimately yields the main sample-complexity bound in \cref{thm:main}.

In \cref{sec:algorithms}, we introduce the two finite-ensemble measurement protocols used throughout the paper. We first define the sequential PGM in \cref{subsec:PGM,subsec:label-walk}. We then introduce the recovery measurement in \cref{subsec:recovery-pgm}. Finally, we establish the near-optimality property of the PGM that will be used to analyze the sequential procedure. In \cref{sec:posterior-localization}, we analyze the sequential PGM and obtain, for finite ensembles, a dimension-independent sample complexity of $O\!\left(\frac{\log^3(M/\delta)}{\varepsilon^2}\right)$. In \cref{sec:recovery-analysis}, we instead analyze the averaged recovery measurement. Although this route initially gives only a weaker conditional mean-bias guarantee, it achieves bias $O\!\left(\sqrt{\frac{\log M}{N}}\right)$ using $N$ copies, simultaneously over all observables.

Finally, \cref{sec:worst-case} converts these finite-prior guarantees into worst-case shadow-tomography protocols. In \cref{subsec:common-finite-decoder}, we first introduce a common finite output alphabet so that all prior-dependent measurements lie in the same compact convex strategy space. In \cref{subsec:finite-state-minimax}, a minimax argument removes the dependence on the prior, and in \cref{subsec:trace-distance-net}, a trace-distance net extends the resulting guarantee from a finite family of states to all states. Applying this framework to the sequential PGM yields the dimension-independent sample complexity $O\!\left(\frac{\log^3(M/\delta)}{\varepsilon^2}\right)$.

For the recovery-measurement route, \cref{subsec:uniform-bias} first applies minimax to convert the finite-prior conditional mean-bias bound into a uniform bounded-bias estimator for all states using $O\!\left(\frac{\log M}{\varepsilon^2}\right)$ copies. Independent repetition followed by empirical averaging then gives the preliminary shadow-tomography bound $O\!\left(\frac{\log M\,\log(M/\delta)}{\varepsilon^4}\right)$. Finally, in \cref{subsec:geometric-refinement}, we use this constant-accuracy procedure as a subroutine in an iterative refinement scheme, which yields our main bound $O\!\left(\frac{\log M\,\log(M/\delta)}{\varepsilon^2}\right)$.

\paragraph{Acknowledgement:} Q.H. is supported by an IBM-Illinois Discovery Accelerator Institute (IIDAI) research grant.

%% file: sections/related_work.tex
\subsection{Relation to previous work}
\label{subsec:related-work}

\textbf{Shadow Tomography.}
The shadow tomography problem was introduced by Aaronson~\cite{aaronson2018shadow}, who gave an algorithm using $\widetilde O(\varepsilon^{-4}\log^4 M\log d)$ copies and established a lower bound of order $\min\{d^2,\log M\}/\varepsilon^2$. B{\u a}descu and O'Donnell subsequently improved the upper bound to $\widetilde O(\varepsilon^{-4}\log^2 M\log d)$, while quantum event learning provides an alternative route to the same asymptotic sample complexity~\cite{badescuodonnell2024,wattsbostanci2022}. Although these results depend only polylogarithmically on the number $M$ of observables, they retain a logarithmic dependence on the Hilbert-space dimension $d$.

A separate line of work seeks dimension-independent sample complexity. Sinha gave a dimension-independent protocol using $O(\sqrt M\log M/\varepsilon^2)$ copies with constant success probability~\cite{sinha2024}. Chen, Li, and Liu obtained the optimal dimension-independent sample complexity $O(\varepsilon^{-2}\log M)$~\cite{chenliliu2024} in high precision regime when $\varepsilon \leq O(d^{-12})$. Pelecanos, Spilecki, and Wright later extended the range in which this optimal rate is achievable to the broader regime $\varepsilon \leq O(1/d)$~\cite{pelecanosspileckiwright2025}. Recently, Chen, O'Donnell, 
Pelecanos, and Wright further improved the dimension-independent sample complexity to $O(\sqrt M/\varepsilon^2)$~\cite{chenOdonnellPelecanosWright2026}.

Shadow tomography has also been studied in the online setting, in which the observables 
arrive sequentially and the protocol must answer each query before seeing the subsequent ones. 
Aaronson and Rothblum developed an online protocol through a connection 
between gentle quantum measurement and differential privacy~\cite{aaronson2019gentlemeasurementquantumstates}. Their protocol uses 
$\widetilde O(\varepsilon^{-8}\log^2 M\log^2 d)$ copies. The protocol of B{\u a}descu 
and O'Donnell also applies in the online setting and improves the sample complexity 
to $\widetilde O(\varepsilon^{-4}\log^2 M\log d)$. Also, Chen, O'Donnell, 
Pelecanos, and Wright further improved the online sample complexity to $O(\min\{\log M\sqrt{\log d}/\varepsilon^3,\sqrt M/\varepsilon^2\})$~\cite{chenOdonnellPelecanosWright2026}.

Several additional lines of work investigate shadow tomography under structural 
assumptions on the observables, the unknown state, or the allowed measurement scheme. 
For structured classes of observables, Huang, Kueng, and Preskill introduced the 
classical-shadows framework, a closely related variant of shadow tomography designed 
to predict many properties of an unknown quantum state from a compact classical 
representation~\cite{huangkuengpreskill2020}. For Pauli observables and restricted 
measurements, King, Gosset, Kothari, and Babbush developed triply efficient protocols 
that are simultaneously efficient in sample complexity, classical computation time, and 
measurement entanglement~\cite{kinggossetkotharibabbush2025}. For thermal states, Chen 
and Gily{\'e}n obtained the optimal $O(\log M)$ sample complexity when suitable access 
to the underlying Hamiltonian is available~\cite{CG26}. Finally, Chen, Gong, and Zhou 
gave instance-optimal characterizations of shadow tomography with one-copy and few-copy 
measurements below an instance-dependent precision threshold~\cite{chengongzhou2026}.
\newline

\textbf{Pretty Good Measurement}. The pretty-good or square-root measurement has its roots in quantum
hypothesis testing
~\cite{belavkin1975,hausladenwootters1994}.  Barnum and Knill proved its
near-optimality for discrimination and recovery tasks
~\cite{barnumknill2002}.  The precise quadratic comparison used here is due
to Mishra, Lami, and Wilde, their general positive-score theorem compares
the optimal gain with the gain of the generalized square-root measurement,
and its Euclidean mean-square specialization is at most twice the optimum~\cite{mishralamiwilde2025}. We include its short
finite-dimensional specialization in
\cref{subsec:quadratic-comparison} so that the constants and support
conventions needed later are transparent. 
\newline

\textbf{Recovery Measurement}
The recovery map viewpoint originated from Petz's characterization of equality in the data processing inequality and quantum sufficiency~\cite{petz1986sufficient, petz1988sufficiency}, which was connected to equality in strong subadditivity and quantum Markov structure~\cite{haydenjozsapetzwinter2004}. Fawzi and Renner later showed that small conditional mutual information implies approximate recoverability~\cite{fawzirenner2015}, which motivates lots of strengthened data processing and explicit recovery results based on the rotated or averaged version of Petz maps~\cite{wilde2015recoverability, sutterfawzirenner2016, sutterbertatomamichel2017, junge2018recovery}. In particular, Junge et al. constructed a universal recovery map, which only depends on the reference state and the corresponding channel, by averaging rotated Petz maps with an explicit probability density~\cite{junge2018recovery}. Our recovery measurement is obtained by specializing the universal recovery map to a classical label register and reading out the recovered label, the recovery map is used as an estimation measurement rather than a procedure of reconstructing the quantum subsystem. 
\newline

\textbf{Minimax theorem}. Finally, the passage from an ensemble-average guarantee to a worst-case
measurement is a minimax argument, rather than an additional tomography
measurement.  We use the finite-dimensional form of Sion's theorem
~\cite{sion1958}, which provides a general framework for solving minimax
problems in quantum information theory.  Broader quantum minimax theorems and least-favourable
priors are developed in~\cite{tanaka2015}.  A fixed finite decoder is
introduced before minimax.  This small ordering detail is essential because
it places every prior-dependent construction in the same compact convex
strategy space.

%% file: sections/algorithms.tex
\section{Finite ensembles and the two measurements}
\label{sec:algorithms}

\subsection{Pretty-good measurement}
\label{subsec:PGM}
Let
\[
  \mathcal F=\{(q_x,\rho_x):x\in\mathcal X\},
  \qquad q_x>0,\qquad \sum_xq_x=1,
\]
be a finite ensemble on a Hilbert space $\mathcal H$. Operationally, the
label $X$ is drawn with probability $q_x$, and multiple identical copies
of the corresponding state $\rho_X$ are prepared. The receiver then
performs a measurement in an attempt to identify $X$. This is known as
the \textit{state discrimination} problem. A commonly used measurement
for this problem is the pretty-good measurement (PGM).

\begin{definition}
  \label{lem:completed-pgm}
  The operators
  \begin{equation}
    G_x
    \coloneqq
    \bar\rho^{-1/2}q_x\rho_x^{\otimes T}\bar\rho^{-1/2}
    +q_x(I-P),
    \qquad x\in\mathcal X,
    \label{eq:completed-pgm}
  \end{equation}
  form the completed pretty-good measurement, where
  \[
    \bar\rho\coloneqq\sum_xq_x\rho_x^{\otimes T}
  \]
  is the average state of the $T$-copy ensemble, $P$ denotes the
  orthogonal projector onto $\supp\bar\rho$, and the inverse is taken on
  $\supp\bar\rho$. The second term in Eq.\eqref{eq:completed-pgm} has zero
  contribution on every state $\rho_x^{\otimes T}$ in the ensemble.
\end{definition}

Zero-weight states may always be removed before applying the definition.
This support convention avoids artificial inverse-eigenvalue assumptions
and will be used without further comment.

\subsection{$r$-round sequential pretty-good measurement}
\label{subsec:label-walk}

For this procedure, write the known ensemble as
$\cE=\{(p_x,\rho_x):x\in\mathcal X\}$, with prior $p_x$.
In this work, we consider a sequential variant of the PGM acting on $T$
copies of the unknown state. We assume that $T=rn$ and partition the
copies into $r$ groups of equal size $n$. We then update the estimate of
the unknown label iteratively by performing a PGM on each group of fresh
copies.

In the first round, we perform the ordinary PGM and obtain the label
estimate $Y_1$. We then perform a PGM on a new group of copies, but
replace the original prior distribution by the posterior distribution
obtained from the previous measurement outcomes. Thus, in each round, we
perform the PGM associated with the updated prior distribution. Specifically, write
\[
  H_{s-1}\coloneqq(Y_1,\ldots,Y_{s-1}),
  \qquad H_0\coloneqq\varnothing,
\]
for the measurement history before round $s$. For a specific history
$h=(y_1,\ldots,y_{s-1})$ obtained from the first $s-1$ rounds, denote the
updated prior distribution and its support by
\begin{equation*}
  p_x^h
  \coloneqq
  \Pr[X=x\mid H_{s-1}=h],
  \qquad
  \mathcal X_h
  \coloneqq
  \{x\in\mathcal X:p_x^h>0\}.
  \end{equation*}
At round $s$, the relevant ensemble is therefore the posterior ensemble
\[
  \mathcal F_h^{(s)}
  \coloneqq
  \{(p_x^h,\rho_x^{\otimes n}):x\in\mathcal X_h\}.
\]
Set
\[
  \bar\rho_h^{(s)}
  \coloneqq
  \sum_{x\in\mathcal X_h}p_x^h\rho_x^{\otimes n},
  \qquad
  \Pi_h^{(s)}
  \coloneqq
  \Pi_{\supp\bar\rho_h^{(s)}}.
\]
The measurement implemented at round $s$ is the completed PGM of this
posterior ensemble.
\begin{equation}
  G_{y\mid h}^{(s)}
  \coloneqq
  \bigl(\bar\rho_h^{(s)}\bigr)^{-1/2}
  p_y^h\rho_y^{\otimes n}
  \bigl(\bar\rho_h^{(s)}\bigr)^{-1/2}
  +p_y^h\bigl(I-\Pi_h^{(s)}\bigr),
  \qquad y\in\mathcal X_h.
  \label{eq:conditional-pgm}
\end{equation}
Specifically, the conditional probability of obtaining the outcome $y$
is
\[
  \Pr[Y_s=y\mid X=x,H_{s-1}=h]
  =
  \Tr\!\left(G_{y\mid h}^{(s)}\rho_x^{\otimes n}\right).
\]
By Bayes' rule, the marginal probability of obtaining the outcome $y$ is
\[
\begin{aligned}
  \Pr[Y_s=y\mid H_{s-1}=h]
  &=
  \sum_{x\in\mathcal X_h}
  p_x^h\Tr\!\left(G_{y\mid h}^{(s)}\rho_x^{\otimes n}\right)
  \\
  &=
  \Tr\!\left(G_{y\mid h}^{(s)}\bar\rho_h^{(s)}\right)
  \\
\end{aligned}
\]

\noindent
After observing $Y_s=y$, the updated history is $h'\coloneqq(h,y)$, and
Bayes' rule gives
\begin{align}
  p_x^{h'}
  &=
  \Pr[X=x\mid H_{s-1}=h,Y_s=y]
  =
  \frac{
    p_x^h\Tr\!\left(G_{y\mid h}^{(s)}\rho_x^{\otimes n}\right)
  }{
    p_y^h
  }.
  \label{eq:posterior-update}
\end{align}
Hence, the posterior distribution at the next round can be derived
entirely from the posterior distribution at the current round and the
observed measurement outcome. Starting from
$p_x^{\varnothing}=p_x$, we repeat this construction for
$s=1,\ldots,r$ and return the final PGM outcome
\[
  J\coloneqq Y_r.
\]
We also keep the complete classical history $H_r$ for the decoder in
Theorem~\ref{thm:recursive-posterior-localization}, maintaining this record
leaves every conditional measurement in Eq.\eqref{eq:conditional-pgm}
unchanged.
The complete algorithm can be written as follows.
\begin{figure}[H]
\centering
\begin{tikzpicture}[x=1cm, y=1cm,
  pgmgate/.style = {draw=mainC!80!black, fill=mainC!10, line width=0.6pt,
                    rounded corners=3pt, minimum width=1.9cm,
                    minimum height=1.0cm, inner sep=1pt},
  card/.style    = {draw=auxC!45, fill=auxC!5, line width=0.5pt,
                    rounded corners=3pt},
  bayes/.style   = {draw=auxC!85!black, fill=white, line width=0.5pt,
                    font=\scriptsize, text=stInk, inner xsep=1.6pt,
                    inner ysep=1.8pt, rounded corners=2.5pt},
  bar/.style     = {fill=mainC!35},
  barX/.style    = {fill=mainC!90!black},
  ymark/.style   = {fill=auxC!90!black},
  cfg/.style     = {->, line width=0.65pt, draw=mainC!75!black},
  chain/.style   = {arr, -{Stealth[length=1.7mm,width=1.4mm]}},
  pics/post/.style 2 args = {code={
    \draw[stInk!50, line width=0.5pt]
      ({-0.5*\cw-0.04},0) -- ({0.5*\cw+0.04},0);
    \foreach \v [count=\i] in {#1}{
      \ifnum\i=\Xtrue \def\bst{barX}\else \def\bst{bar}\fi
      \fill[\bst] ({-0.5*\cw+(\i-1)*\bp+0.5*(\bp-\bw)},0)
                  rectangle ++(\bw,{\bh*\v});
      \ifnum\i=#2
        \fill[ymark] ({-0.5*\cw+(\i-0.5)*\bp},-0.045)
                     -- ++(-0.07,-0.12) -- ++(0.14,0) -- cycle;
      \fi
    }
  }},
]
\colorlet{mainC}{stPurple}   
\colorlet{auxC}{stGray}      

\def\pitch{2.75} 
\def\cw{1.17}    
\def\bp{0.09}    
\def\bw{0.064}   
\def\bh{1.95}    
\def\Xtrue{9}    
\def\kw{0.66}    
\def\yb{-0.40}   
\def\yc{0.07}    
\def\yg{2.2}     

\def\xa{0.8}
\pgfmathsetmacro\xb{\xa+\pitch}
\pgfmathsetmacro\xc{\xb+\pitch}
\pgfmathsetmacro\xd{\xc+\pitch+0.8}
\pgfmathsetmacro\xe{\xd+\pitch}

\foreach \x/\nm/\lab in {\xa/c0/{\text{prior }p}, \xb/c1/{p^{H_1}},
                         \xc/c2/{p^{H_2}}, \xd/c4/{p^{H_{r-1}}},
                         \xe/c5/{p^{H_r}}}{
  \draw[card] (\x-\kw,\yb-0.32) rectangle (\x+\kw,\yb+1.30);
  \node[inner sep=0pt] (\nm) at (\x,\yc)
       [minimum width=2*\kw cm, minimum height=1.62cm] {};
  \node[lbl, anchor=north west, inner sep=2pt]
       at (\x-\kw,\yb+1.30) {$\lab$};
}
\pic at (\xa,\yb) {post={0.075,0.086,0.095,0.101,0.103,0.101,0.095,0.086,0.075,0.063,0.050,0.039,0.029}{0}};
\pic at (\xb,\yb) {post={0.000,0.000,0.000,0.004,0.032,0.115,0.239,0.293,0.213,0.086,0.017,0.001,0.000}{8}};
\pic at (\xc,\yb) {post={0.000,0.000,0.000,0.000,0.001,0.016,0.114,0.325,0.369,0.155,0.020,0.000,0.000}{9}};
\pic at (\xd,\yb) {post={0.000,0.000,0.000,0.000,0.000,0.000,0.002,0.083,0.484,0.397,0.034,0.000,0.000}{9}};
\pic at (\xe,\yb) {post={0.000,0.000,0.000,0.000,0.000,0.000,0.002,0.102,0.559,0.323,0.014,0.000,0.000}{9}};

\foreach \x/\s/\h/\c in {\xa/1/\varnothing/c0, \xb/2/H_1/c1,
                         \xc/3/H_2/c2, \xd/r/H_{r-1}/c4}{
  \node[pgmgate] (g\s) at (\x,\yg) {};
  \pic[scale=1.1] at ($(g\s.center)+(-0.63,-0.02)$) {meter};
  \node[font=\scriptsize, text=mainC!80!black, anchor=base]
       at ($(g\s.center)+(0.27,0.16)$) {PGM};
  \node[font=\small, text=stInk, anchor=base]
       at ($(g\s.center)+(0.27,-0.31)$) {$G^{(\s)}_{\cdot\mid \h}$};
  \draw[cfg] (\c.north) -- (g\s.south);
  \foreach \dx in {-0.27,-0.09,0.09,0.27}
    \draw[qwire] ($(g\s.north)+(\dx,0)$) -- ++(0,0.5);
  \draw[brace] ($(g\s.north)+(-0.34,0.55)$) -- ($(g\s.north)+(0.34,0.55)$)
       node[lbl, midway, above=4pt] {$\rho_X^{\otimes n}$};
  \node[bayes] (b\s) at (\x+0.5*\pitch,\yc) {Bayes};
  \draw[cwire] (g\s.east) -| (b\s.north)
       node[lbl, pos=0.3, above=1pt] {$Y_{\s}$};
  \draw[chain] (\c.east) -- (b\s.west);
}
\draw[chain] (b1.east) -- (c1.west);
\draw[chain] (b2.east) -- (c2.west);
\draw[line width=0.7pt, draw=stInk!75] (b3.east) -- ++(0.3,0) coordinate (e1);
\node[lbl, anchor=west, inner sep=1pt] (dots) at (e1) {$\cdots$};
\draw[chain] (dots.east) -- (c4.west);
\draw[chain] (br.east) -- (c5.west);
\node[lbl] at ($(b3 |- g3.east)!0.5!(gr.west)$) {$\cdots$};

\node[draw=auxC!85!black, fill=auxC!9, line width=0.6pt, rounded corners=3pt,
      align=center, font=\small, text=stInk, inner xsep=5pt, inner ysep=4pt,
      minimum height=1.0cm] (dec) at (\xe+0.25,\yg)
     {decoder\\ $D(H_r)\in[0,1]^M$};

\draw[arr] (c5.north) -- (c5.north |- dec.south)
     node[lbl, midway, right=1pt, inner sep=1pt]
     {$H_r=(Y_1,\dots,Y_r)$};

\end{tikzpicture}
\caption{Sequential PGM}
\end{figure}
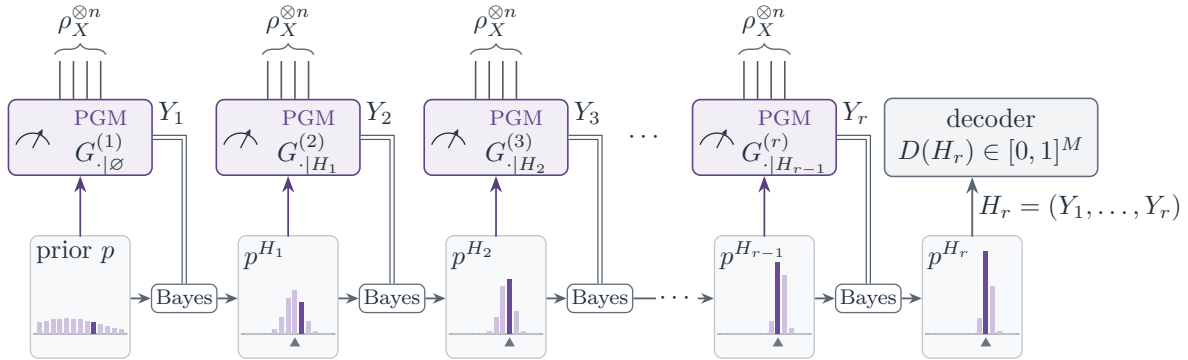

\begin{algorithm}[H]
  \caption{The $r$-round sequential posterior PGM}
  \label{alg:sequential-pgm}
  \begin{algorithmic}[1]
    \Require Ensemble $\cE=\{(p_x,\rho_x):x\in\mathcal X\}$, integers
      $r,n\geq1$, and $r$ groups $B_1,\ldots,B_r$, each containing $n$ fresh
      copies of the unknown state $\rho_X$
    \Ensure The complete history $H_r\in\mathcal X^r$ and final ensemble index $Y_r\in\mathcal X$
    \State $h\gets\varnothing$ and $p_x^h\gets p_x$ for every
      $x\in\mathcal X$
    \For{$s=1,\ldots,r$}
      \State $\mathcal X_h\gets\{x\in\mathcal X:p_x^h>0\}$
      \State $\displaystyle
        \bar\rho_h^{(s)}\gets
        \sum_{x\in\mathcal X_h}p_x^h\rho_x^{\otimes n}$
      \State $\Pi_h^{(s)}\gets\Pi_{\supp\bar\rho_h^{(s)}}$
      \State Construct the completed PGM
        $\{G_{y\mid h}^{(s)}:y\in\mathcal X_h\}$ according to
        Eq.\eqref{eq:conditional-pgm}
      \State Measure $B_s$ with this PGM and denote the outcome by $Y_s$
      \State $h'\gets(h,Y_s)$
      \ForAll{$x\in\mathcal X$}
        \State $\displaystyle
          p_x^{h'}\gets
          \frac{
            p_x^h\Tr\!\left(
              G_{Y_s\mid h}^{(s)}\rho_x^{\otimes n}
            \right)
          }{p^h_{Y_s}}$
      \EndFor
      \State $h\gets h'$
    \EndFor
    \State \Return $(H_r,Y_r)$, where $H_r=h$
  \end{algorithmic}
\end{algorithm}

The construction depends only on the ensemble since it contains no information about the observable $E$. In general, it is not equivalent to the
ordinary PGM of the $rn$-copy ensemble, because the prior distribution
is recomputed after each observed outcome.

For a positive-probability history, set $G_{y\mid h}^{(s)}=0$ when
$y\notin\mathcal X_h$.  At a zero-probability history choose a full
$\mathcal X$-outcome POVM, which can define every summand in
Eq.\eqref{eq:flattened-sequential-pgm}.  Moreover, at a positive probability history,
$\Tr(G_{y\mid h}^{(s)}\bar\rho_h^{(s)})=p_y^h$ by
Eq.\eqref{eq:conditional-pgm}, which supplies the denominator in
Eq.\eqref{eq:posterior-update}. And for a history that occurs with probability zero, the corresponding
conditional PGM can be defined arbitrarily. This convention has no
effect on the operational procedure or any of its outcome probabilities.

\begin{remark}
  \label{prop:flatten}
  After discarding the rest of the history and retaining only $J$, the sequential procedure in \cref{alg:sequential-pgm} is equivalent to
  the $\lvert\mathcal X\rvert$-outcome POVM
  \begin{equation}
    M_j^{[r]}
    \coloneqq
    \sum_{y_1,\ldots,y_{r-1}\in\mathcal X}
      G_{y_1\mid\varnothing}^{(1)}
      \otimes
      G_{y_2\mid y_1}^{(2)}
      \otimes\cdots\otimes
      G_{j\mid(y_1,\ldots,y_{r-1})}^{(r)},
    \qquad j\in\mathcal X,
    \label{eq:flattened-sequential-pgm}
  \end{equation}
  on $\mathcal H^{\otimes rn}$. More precisely, for every
  $x,j\in\mathcal X$,
  \[
    \Pr[J=j\mid X=x]
    =
    \Tr\!\left(M_j^{[r]}\rho_x^{\otimes rn}\right).
  \]
\end{remark}

\input{sections/recovery_algorithm}

\subsection{Quadratic labels and the factor-two theorem}
\label{subsec:quadratic-comparison}

We now return to the one-block ensemble $\mathcal F$ introduced at the
beginning of this section and take $T=1$. Thus,
\[
  \bar\rho\coloneqq\sum_xq_x\rho_x.
\]
The key feature of the PGM is that its operators do not depend on the
scalar quantity that one ultimately wishes to estimate. Assign an
arbitrary real label $z_x$ to each state in $\mathcal F$. A general
scalar estimation strategy consists of a POVM
$\{N_w:w\in\mathcal W\}$ and a reported value
$\widehat z_w \in \mathbb R$ for each measurement outcome. Its mean-square risk is
\begin{equation}
  \sum_{x\in\mathcal X}\sum_{w\in\mathcal W}
    q_x\Tr(N_w\rho_x)(z_x-\widehat z_w)^2.
  \label{eq:scalar-risk}
\end{equation}

The following finite-dimensional form of the Personick estimator
\cite{personick1971,helstrom1976} is included to fix the normalization.

\begin{proposition}
  \label{prop:personick}
  The minimum of Eq.\eqref{eq:scalar-risk} over all scalar strategies is
  \begin{equation*}
    \sum_xq_xz_x^2
    -
    \Tr\!\left(
      L\sum_xq_xz_x\rho_x
    \right),
      \end{equation*}
  where $L$ is the unique Hermitian solution of the Sylvester equation
  \begin{equation}
    \bar\rho L+L\bar\rho
    =
    2\sum_xq_xz_x\rho_x
    \label{eq:sylvester}
  \end{equation}
  on $\supp(\bar\rho)$. The optimal strategy is obtained by extending
  $L$ by zero on $\ker\bar\rho$, measuring its full spectral
  projection-valued measurement (PVM), and
  reporting the corresponding eigenvalues.
\end{proposition}

\begin{proof}
  Given the spectral decomposition
  \[
    \bar\rho=\sum_ad_a\ketbra{a}{a}
  \]
  on $\supp\bar\rho$, define the operator $L$ by
  \begin{equation*}
    \bra{a}L\ket{b}
    =
    \frac{
      2\bra{a}\bigl(\sum_xq_xz_x\rho_x\bigr)\ket{b}
    }{d_a+d_b}.
  \end{equation*}
  This is a solution of Eq.\eqref{eq:sylvester} by construction. Moreover,
  every $d_a$ appearing on $\supp\bar\rho$ is strictly positive.
  Therefore, $\bar\rho$ and $-\bar\rho$ have disjoint spectra on
  $\supp\bar\rho$, so the solution of Eq.\eqref{eq:sylvester} on this
  subspace is unique.
  Using Eq.\eqref{eq:sylvester} and $\sum_wN_w=I$, a direct expansion of the
  squared loss gives
  \begin{align}
    \sum_{x,w}q_x\operatorname{Tr}(N_w\rho_x)(z_x-\widehat z_w)^2
    &=
    \sum_xq_xz_x^2
    -2\sum_{w}\operatorname{Tr}\left(
      \widehat{z}_wN_w
      \left(\sum_x q_xz_x\rho_x\right)
    \right)
    +\sum_w\operatorname{Tr}(\widehat z_w^2N_w\bar\rho)
    \notag\\
    &=
    \sum_xq_xz_x^2
    -\sum_{w}\operatorname{Tr}\left(
      \widehat{z}_wN_w(\bar\rho L+L\bar\rho)
    \right)
    +\sum_w\operatorname{Tr}(\widehat z_w^2N_w\bar\rho)
    \notag\\
    &=
    \sum_xq_xz_x^2
    -\sum_{w}\widehat{z}_w\operatorname{Tr}\left(
      \bar\rho^{1/2}N_wL\bar\rho^{1/2}
    \right)
    \notag\\
    &\quad
    -\sum_{w}\widehat{z}_w\operatorname{Tr}\left(
      \bar\rho^{1/2}LN_w\bar\rho^{1/2}
    \right)
    +\sum_w\widehat z_w^2\operatorname{Tr}\left(
      \bar\rho^{1/2}N_w\bar\rho^{1/2}
    \right)
    \notag\\
    &=
    \sum_xq_xz_x^2
    -\sum_w
      \operatorname{Tr}\!\left[
        \bar\rho^{1/2}L
        N_w
        L\bar\rho^{1/2}
      \right]
    \notag\\
    &\quad+
    \sum_w
      \operatorname{Tr}\!\left[
        \bar\rho^{1/2}(\widehat z_wI-L)
        N_w
        (\widehat z_wI-L)\bar\rho^{1/2}
      \right].
    \label{eq:personick-completion}
  \end{align}
  Each summand in the final term is nonnegative because
  \[
    \operatorname{Tr}\!\left[
      \bar\rho^{1/2}(\widehat z_w I-L)
      N_w
      (\widehat z_w I-L)\bar\rho^{1/2}
    \right]
    =
    \left\|
      N_w^{1/2}(\widehat z_w I-L)\bar\rho^{1/2}
    \right\|_2^2
    \geq 0.
  \]
  Hence, the entire final term is nonnegative. Furthermore, using
  $\sum_wN_w=I$, we obtain
  \[
  \begin{aligned}
    \sum_w
      \operatorname{Tr}\!\left[
        \bar\rho^{1/2}L
        N_w
        L\bar\rho^{1/2}
      \right]
    &=
    \operatorname{Tr}(L\bar\rho L)
    \\
    &=
    \frac{1}{2}
    \operatorname{Tr}\!\left[
      L(\bar\rho L+L\bar\rho)
    \right]
    \\
    &=
    \operatorname{Tr}\!\left(
      L\sum_xq_xz_x\rho_x
    \right).
  \end{aligned}
  \]
  Substituting this identity into Eq.\eqref{eq:personick-completion} gives
  \[
    \sum_{x,w}q_x\operatorname{Tr}(N_w\rho_x)
    (z_x-\widehat z_w)^2
    \geq
    \sum_xq_xz_x^2
    -
    \Tr\!\left(
      L\sum_xq_xz_x\rho_x
    \right).
  \]
  If the measurement is the spectral PVM of $L$ and the reported values
  are the corresponding eigenvalues, every summand in the final term of
  Eq.\eqref{eq:personick-completion} vanishes. Thus, the lower bound is
  attained.
\end{proof}

\begin{theorem}
  \label{thm:factor-two}
  This is the finite-dimensional quadratic specialization of
  \cite[Theorem~3.6]{mishralamiwilde2025}. If the PGM outcome $Y$ is
  decoded as $z_Y$, then
  \[
    \sum_{x,y}
      q_x\operatorname{Tr}(G_y\rho_x)(z_x-z_y)^2
    \leq
    2
    \inf_{\{N_w\},\{\widehat z_w\}}
    \sum_{x,w}
      q_x\operatorname{Tr}(N_w\rho_x)(z_x-\widehat z_w)^2.
  \]
  The PGM is determined entirely by $\mathcal F$ and is independent of
  the labels $z$.
\end{theorem}

\begin{proof}
  The symmetry and marginal identities of the PGM-induced joint
  distribution give
  \begin{align}
    \sum_{x,y}q_x\Tr(G_y\rho_x)(z_x-z_y)^2
    &=
    \sum_{x,y}q_x\Tr(G_y\rho_x)z_x^2
    -
    2\sum_{x,y}q_x\Tr(G_y\rho_x)z_xz_y
    +
    \sum_{x,y}q_x\Tr(G_y\rho_x)z_y^2
    \notag\\
    &=
    2\sum_xq_xz_x^2
    -
    2\sum_{x,y}q_x\Tr(G_y\rho_x)z_xz_y
    \notag\\
    &=
    2\sum_xq_xz_x^2
    -
    2\Tr\!\left[
      \left(\sum_xq_xz_x\rho_x\right)
      \bar\rho^{-1/2}
      \left(\sum_yq_yz_y\rho_y\right)
      \bar\rho^{-1/2}
    \right].
    \label{eq:pgm-risk-formula}
  \end{align}
  In the eigenbasis
  $\bar\rho=\sum_ad_a\ketbra{a}{a}$ on $\supp\bar\rho$, we have
  \[
    \operatorname{Tr}\!\left[
      \left(\sum_xq_xz_x\rho_x\right)
      \bar\rho^{-1/2}
      \left(\sum_yq_yz_y\rho_y\right)
      \bar\rho^{-1/2}
    \right]
    =
    \sum_{a,b}
    \frac{
      \abs{
        \bra{a}
        \left(\sum_xq_xz_x\rho_x\right)
        \ket{b}
      }^2
    }{\sqrt{d_ad_b}}.
  \]
  Moreover,
  \[
    \operatorname{Tr}\!\left(
      L\sum_xq_xz_x\rho_x
    \right)
    =
    \sum_{a,b}
    \frac{
      2\abs{
        \bra{a}
        \left(\sum_xq_xz_x\rho_x\right)
        \ket{b}
      }^2
    }{d_a+d_b}.
  \]
  We have,
  \[
    \frac{1}{\sqrt{d_ad_b}}
    \geq
    \frac{2}{d_a+d_b}.
  \]
  Substituting this inequality into the preceding expressions and
  applying Proposition~\ref{prop:personick} proves the theorem.
\end{proof}

The general theorem of Mishra, Lami, and Wilde
\cite{mishralamiwilde2025} applies to continuous ensembles,
infinite-dimensional systems, and substantially more general score
functions. \Cref{thm:factor-two} is precisely its quadratic consequence
in the present finite-dimensional setting.

%% file: sections/recovery_algorithm.tex
\subsection{The averaged-recovery label measurement}
\label{subsec:recovery-pgm}

We now use the same finite ensemble $\cE=\{(p_x,\rho_x):x\in\mathcal X\}$
with its original prior.  All states and probabilities are known, while
$X$ is unknown.  Set
\[
  \theta_j(\rho)\coloneqq\Tr(E_j\rho),
  \qquad
  \vartheta(\rho)\coloneqq
  \bigl(\operatorname{Tr}(E_1\rho),\ldots,\operatorname{Tr}(E_M\rho)\bigr).
\]
The same response vector is used for both algorithms.  Fix an integer
budget $N\geq1$, and for $0\leq t<N$ set
\[
  B_t\coloneqq\sum_xp_x\rho_x^{\otimes t},
  \qquad P_t\coloneqq\Pi_{\supp B_t},
  \qquad
  \beta_0(u)\coloneqq\frac{\pi}{2(\cosh(\pi u)+1)}.
\]
Define
\begin{equation}
  D_{y,t}
  \coloneqq
  p_y\int_{\mathbb R}\beta_0(u)
  B_t^{-(1+iu)/2}\rho_y^{\otimes t}B_t^{-(1-iu)/2}\dd u
  +p_y(I-P_t).
  \label{eq:recovery-effects}
\end{equation}
Complex powers here are computed on $\supp B_t$ and extended by zero.
The integral converges in finite dimension.  Indeed the imaginary
powers are unitary on this support, so its integrand in norm is bounded
by $\beta_0(u)\|B_t^{-1/2}\|^2\|\rho_y^{\otimes t}\|$.
The function $\beta_0$ is an even probability density, since
\[
  \beta_0(u)=\frac{\pi}{4\cosh^2(\pi u/2)},
  \qquad
  \int_{\mathbb R}\beta_0(u)\dd u
  =\frac12[\tanh(\pi u/2)]_{-\infty}^{\infty}=1.
\]
The specialization of that recovery
map is proved in Section~\ref{sec:recovery-analysis}.

\begin{algorithm}[H]
  \caption{The recovery PGM}
  \label{alg:recovery-pgm}
  \begin{algorithmic}[1]
    \Require Ensemble $\cE=\{(p_x,\rho_x):x\in\mathcal X\}$,
      integer $N\geq1$, and $N$ fresh copies of $\rho_X$
    \Ensure An ensemble index $Y\in\mathcal X$
    \State Choose $t$ uniformly from $\{0,\ldots,N-1\}$
    \If{$t=0$}
      \State Sample $Y$ from the law $p$
    \Else
      \State Construct $\{D_{y,t}:y\in\mathcal X\}$
      \State Measure the first $t$ copies with this POVM and denote the outcome by $Y$
    \EndIf
    \State Ignore the remaining copies
    \State \Return $Y$
  \end{algorithmic}
\end{algorithm}

\begin{figure}[H]
\centering
\begin{tikzpicture}[x=1cm, y=1cm,
  copy/.style  = {draw=mainC!75!black, fill=mainC!11, rounded corners=2pt,
                  line width=0.6pt, minimum width=\tw cm, minimum height=\tht cm,
                  inner sep=0pt, font=\small, text=stInk},
  idle/.style  = {copy, draw=auxC!55, fill=auxC!5, text=auxC,
                  dash pattern=on 2pt off 1.2pt},
  ground/.style= {line width=0.6pt, draw=auxC!80},
]
\colorlet{mainC}{stPurple}   
\colorlet{auxC}{stGray}      

\def\Ncp{10}      
\def\tcut{5}      
\def\tw{0.70}     
\def\tht{0.56}    
\def\pitch{0.82}  
\def\gap{0.36}    
\def\ty{4.35}     
\def\gy{2.62}     
\def\ix{8.95}     
\def\iX{15.62}    
\def\iy{1.95}     
\def\iY{5.62}     

\foreach \i in {1,...,\Ncp}{
  \pgfmathsetmacro{\xx}{(\i-0.5)*\pitch + (\i>\tcut ? \gap : 0)}
  \ifnum\i>\tcut
    \node[idle] (c\i) at (\xx,\ty) {$\rho_X$};
  \else
    \node[copy] (c\i) at (\xx,\ty) {$\rho_X$};
  \fi
}
\coordinate (cut) at ($(c\tcut.east)!0.5!(c6.west)$);
\draw[stInk!70, line width=0.7pt, dash pattern=on 3pt off 2pt]
  ($(cut)+(0,-0.72)$) -- ($(cut)+(0,1.25)$);
\draw[brace, draw=mainC!80!black] (c1.north west) -- (c\tcut.north east)
  node[midway, above=4pt, font=\footnotesize, text=mainC!85!black] {$t$ copies};
\draw[brace, draw=auxC!80] (c6.north west) -- (c\Ncp.north east)
  node[midway, above=4pt, font=\footnotesize, text=auxC] {$N-t$ copies};

\begin{scope}[shift={($(cut)+(-1.45,1.55)$)}]
  \draw[stInk!75, fill=white, line width=0.6pt, rounded corners=1.6pt]
    (-0.19,-0.19) rectangle (0.19,0.19);
  \foreach \p in {(-0.095,-0.095),(0.095,0.095),(-0.095,0.095),(0.095,-0.095),(0,0)}
    \fill[stInk] \p circle (0.032);
\end{scope}
\node[lbl, anchor=west] (tlab) at ($(cut)+(-1.18,1.55)$)
  {$t\sim\mathrm{Unif}\{0,\ldots,N-1\}$};

\node[draw=mainC!80!black, fill=mainC!9, line width=0.7pt,
      rounded corners=2.5pt, minimum height=0.9cm,
      minimum width={(\tcut-1)*\pitch cm + \tw cm + 0.2cm}]
  (G) at ($(c1)!0.5!(c\tcut) + (0,\gy-\ty)$) {};
\node[lbl, anchor=center] at ($(G.center)+(-0.27,0)$)
  {measure $\{D_{y,t}\}_{y\in\mathcal X}$};
\pic[scale=1.15] at ($(G.east)+(-0.34,-0.05)$) {meter};
\foreach \i in {1,...,\tcut}
  \draw[qwire] (c\i.south) -- (c\i.south |- G.north);

\foreach \i in {6,...,\Ncp}{
  \draw[ground] (c\i.south) -- ++(0,-0.32) coordinate (gd\i);
  \draw[ground] ($(gd\i)+(-0.14,0)$) -- ($(gd\i)+(0.14,0)$);
  \draw[ground] ($(gd\i)+(-0.09,-0.06)$) -- ($(gd\i)+(0.09,-0.06)$);
  \draw[ground] ($(gd\i)+(-0.04,-0.12)$) -- ($(gd\i)+(0.04,-0.12)$);
}

\node[box, draw=auxC!80!black, anchor=east] (out) at ($(c\Ncp.east |- G)$)
  {$\vartheta(\rho_Y)\in[0,1]^M$};
\draw[cwire, -{Stealth[length=2mm,width=1.6mm]}] (G.east) -- (out.west)
  node[midway, above=1pt, lbl] {$Y$};

\def\ry{3.78}     
\def\rx{11.0}     
\def\qx{13.15}    
\def\us{0.55}     
\draw[mainC!50, fill=mainC!3, line width=0.6pt, rounded corners=3pt]
  (\ix,\iy) rectangle (\iX,\iY);
\node[font=\small, text=mainC!85!black, anchor=north west]
  at (\ix+0.12,\iY-0.1) {$\{D_{y,t}\}_y$:};
\node[draw=mainC!80!black, fill=mainC!10, line width=0.6pt, rounded corners=2.5pt,
      minimum width=1.3cm, minimum height=0.75cm, font=\small, text=stInk]
  (R) at (\rx,\ry) {$B_t^{\,iu/2}$};
\node[draw=mainC!80!black, fill=mainC!10, line width=0.6pt, rounded corners=2.5pt,
      minimum width=1.6cm, minimum height=0.75cm] (Q) at (\qx,\ry) {};
\pic[scale=1.1] at ($(Q.center)+(-0.43,-0.04)$) {meter};
\node[font=\scriptsize, text=mainC!80!black] at ($(Q.center)+(0.27,0)$) {PGM};
\draw[qwire] (\ix+0.3,\ry) -- (R.west);
\node[lbl, anchor=south] at (\ix+0.8,\ry+0.02) {$\rho_X^{\otimes t}$};
\draw[qwire] (R.east) -- (Q.west);
\draw[cwire, -{Stealth[length=2mm,width=1.6mm]}] (Q.east) -- ++(0.85,0)
  node[right, lbl, inner sep=2pt] {$Y$};
\begin{axis}[at={(\rx cm-1.0cm, \ry cm+0.72cm)}, anchor=south west,
    scale only axis, width=2.0cm, height=0.62cm,
    xmin=-3.2, xmax=3.2, ymin=0, ymax=0.85, axis lines=none, clip=false]
  \addplot[domain=-3:3, samples=91, draw=mainC!85!black, fill=mainC!22,
           line width=0.8pt] {pi/(4*cosh(pi*x/2)^2)} \closedcycle;
  \draw[stInk!70, line width=0.5pt] (axis cs:-3.2,0) -- (axis cs:3.2,0);
  \coordinate (us) at (axis cs:\us,0);
  \coordinate (bl) at (axis cs:3.2,0.1);
\end{axis}
\draw[cwire] (us) -- (us |- R.north);
\fill[mainC!85!black] (us) circle (0.055);
\node[font=\footnotesize, text=mainC!85!black, anchor=west] at ($(bl)+(0.08,0)$)
  {$u\sim\beta_0$};
\draw[draw=auxC!45, fill=auxC!5, line width=0.5pt, rounded corners=3pt]
  ({(\rx+\qx)/2-0.66},\iy+0.12) rectangle ({(\rx+\qx)/2+0.66},\iy+1.12);
\coordinate (card) at ({(\rx+\qx)/2},\iy+0.62);
\node[lbl, anchor=north west, inner sep=2pt] at ({(\rx+\qx)/2-0.66},\iy+1.12)
  {prior $p$};
\draw[stInk!50, line width=0.5pt] ({(\rx+\qx)/2-0.625},\iy+0.3) -- ({(\rx+\qx)/2+0.625},\iy+0.3);
\foreach \v [count=\i] in {0.075,0.086,0.095,0.101,0.103,0.101,0.095,0.086,0.075,0.063,0.050,0.039,0.029}
  \fill[mainC!35] ({(\rx+\qx)/2-0.585+(\i-1)*0.09+0.013},\iy+0.3) rectangle ++(0.064,{1.95*\v});
\draw[->, line width=0.65pt, draw=mainC!75!black]
  ({(\rx+\qx)/2-0.66},\iy+0.62) -| (R.south);
\draw[->, line width=0.65pt, draw=mainC!75!black]
  ({(\rx+\qx)/2+0.66},\iy+0.62) -| (Q.south);
\end{tikzpicture}
\caption{Recovery PGM}
\end{figure}
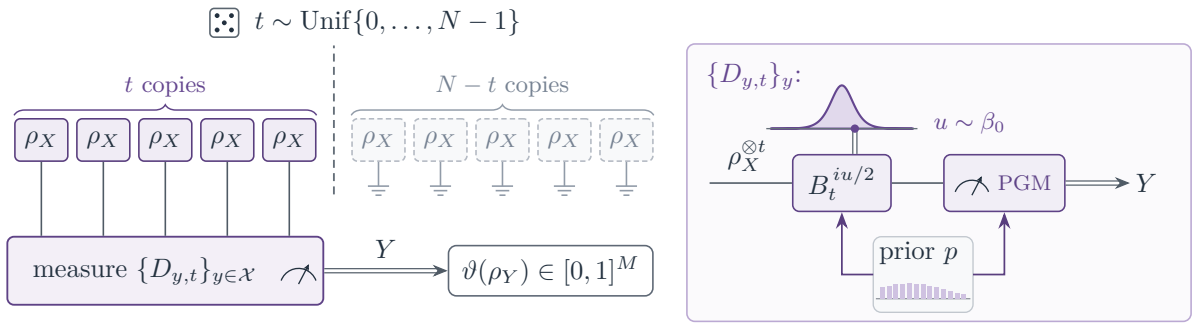

The term recovery PGM denotes this averaged recovery label
measurement, its raw decoder reports $\vartheta(\rho_Y)$.  The
integrated effects are
\begin{equation}
  \overline D_y
  \coloneqq
  \frac1N\sum_{t=0}^{N-1}
  D_{y,t}\otimes I_{\mathcal H^{\otimes(N-t)}}.
  \label{eq:averaged-recovery-povm}
\end{equation}
Lemma~\ref{lem:recovery-normalization} proves that the effects for each $t$ sum
to the identity, thus the effects in
Eq.\eqref{eq:averaged-recovery-povm} are positive and sum to the identity
on $N$ copies.  They have the finite output alphabet $\mathcal X$,
the integral in Eq.\eqref{eq:recovery-effects} is part of an effect,
and does not add an output label.

For clarity, the relation to ordinary PGM can be written on any fixed
ensemble $\{(q_y,\sigma_y)\}$ on one common input space.  Set
$B=\sum_yq_y\sigma_y$, let $P=\Pi_{\supp B}$, and on its support put
\[
  G_y=q_yB^{-1/2}\sigma_yB^{-1/2},
  \qquad U_u=B^{-iu/2}.
\]
Then
\begin{align}
  U_uG_yU_u^\dagger
  &=q_y B^{-iu/2}B^{-1/2}\sigma_yB^{-1/2}B^{iu/2}
    \notag\\
  &=q_y B^{-(1+iu)/2}\sigma_yB^{-(1-iu)/2},
    \notag
\end{align}
and
\begin{align}
  D_y
  &=\int_{\mathbb R}\beta_0(u)U_uG_yU_u^\dagger\dd u
    \notag\\
  &=q_y\int_{\mathbb R}\beta_0(u)
       B^{-(1+iu)/2}\sigma_yB^{-(1-iu)/2}\dd u.
  \label{eq:pgm-recovery-identity}
\end{align}
Extend $U_u$ by the identity on $\ker B$, and complete both $G_y$ and
$D_y$ by adding $q_y(I-P)$.  Since $\int\beta_0=1$, the same identity holds
for these completed effects.

For a fixed commuting ensemble, choose a common eigenbasis $w$, both completed measurements have
\[
  \langle w|G_y|w\rangle =\frac{q_y\langle w|\sigma_y|w\rangle}{\sum_yq_y\langle w|\sigma_y|w\rangle}.
\]
Consequently one can measure $w$ and sample $y$ from this posterior,
on $\sum_yq_y\langle w|\sigma_y|w\rangle=0$ both completions sample $q$, this is a representation of
the two measurements on that fixed ensemble and space.
Algorithm~\ref{alg:sequential-pgm} uses a different posterior on each
fresh block, Algorithm~\ref{alg:recovery-pgm} uses the original prior
on one uniformly chosen prefix.

%% file: sections/posterior_localization.tex
\section{Posterior localization and high-confidence estimation}
\label{sec:posterior-localization}

We store the full
measurement history and recursively change the scalar labels to which \Cref{thm:factor-two} is applied. 
The physical measurement remains the posterior-refocused PGM of \Cref{subsec:label-walk}, only the classical
post-processing is changed.  For a history $h=(y_1,\ldots,y_r)$, the corresponding full-history
POVM element is
\begin{equation}
\widetilde M^{[r]}_{h}
:=G^{(1)}_{y_1\mid\varnothing}\otimes G^{(2)}_{y_2\mid y_1}\otimes\cdots\otimes
G^{(r)}_{y_r\mid(y_1,\ldots,y_{r-1})}.
\label{eq:full-history-povm}
\end{equation}
Summing the last factor over $y_r$ gives $I$, and repeating this
operation for $y_{r-1},\ldots,y_1$ gives
$\sum_h\widetilde M_h^{[r]}=I$.  These operators form a POVM, and \eqref{eq:flattened-sequential-pgm} is their coarse-graining according to the final label
$y_r$.

At a fixed center $a$, the signed
power $(u-a)|u-a|^{s-1}/s!$ has squared magnitude
$|u-a|^{2s}/(s!)^2$, so a quadratic bound on this coordinate controls a
$2s$-th moment of the original error, so, its derivative is the
magnitude of the preceding signed power. This derivative relation
allows the leading sampling error on a fresh block to be weighted by
a residual which has been controlled in the preceding round. The posterior
center changes with the observations, however, recentering a power
about a new input value would also change its derivatives. Our
construction separates these two operations, first, we obtain the next
coordinate by integrating the magnitude of the current residual, and
then subtract its posterior mean after the next outcome. This
subtraction adapts the centering to the new posterior while leaving
the derivative relation unchanged. Thus, the resulting functions are dependent to full history of measurements, and the following estimate ensures that these bounds are still controlling high powers of the distance in the original expectation coordinate.

\subsection{Recursively centered coordinates}

Fix an observable $0\preceq E\preceq I$ and write
\[
\theta_x:=\operatorname{Tr}(E\rho_x),\qquad \Theta:=\theta_X.
\]
Thus $\Theta\in[0,1]$ is a classical random variable under every posterior.  Set $f_0(u):=1$.
For a positive-probability history $h_s=(y_1,\ldots,y_s)$, with prefix
$h_{s-1}=(y_1,\ldots,y_{s-1})$, define recursively
\begin{equation*}
F_{s-1}^{h_{s-1}}(u)
:=\int_0^u\left|f_{s-1}^{h_{s-1}}(v)\right|\,dv,
\qquad
c_s(h_s)
:=\mathbb E\!\left[F_{s-1}^{h_{s-1}}(\Theta)\mid H_s=h_s\right],
\qquad
f_s^{h_s}(u)
:=F_{s-1}^{h_{s-1}}(u)-c_s(h_s)
\end{equation*}
Zero-probability histories are completed arbitrarily, these are just labels only used for analysis, and are not supplied to the physical PGM.  By construction,
\begin{equation*}
\mathbb E\!\left[f_s^{H_s}(\Theta)\mid H_s\right]=0,
\qquad
\mathbb E\!\left[f_s^{H_s}(\Theta)^2\mid H_s\right]
=\operatorname{Var}\!\left(F_{s-1}^{H_{s-1}}(\Theta)\mid H_s\right).
\end{equation*}
Moreover, $(F_{s-1}^{h_{s-1}})'=|f_{s-1}^{h_{s-1}}|$.  Hence, for a small perturbation $z$,
\[
F_{s-1}(u+z)-F_{s-1}(u)
\approx z\,|f_{s-1}(u)| + O(z^2).
\]
Thus the new coordinate is flat where the previous residual is already small.

We next record the two deterministic estimates used below.  Fix a realized history path and
suppress its prefixes from the notation, a direct induction from
$f_j(v)-f_j(u)=\int_u^v|f_{j-1}(w)|\,dw$ gives, whenever $u,u+z\in[0,1]$,
\begin{equation}
\left|F_{s-1}(u+z)-F_{s-1}(u)\right|
\le
\sum_{\ell=1}^{s}\frac{|z|^\ell}{\ell!}\,|f_{s-\ell}(u)|.
\label{eq:triangular-increment}
\end{equation}
The same recursion also prevents the transformed coordinate from collapsing a long interval,
for every $0\le u<v\le 1$,
\begin{equation}
f_s(v)-f_s(u)
\ge
\frac{(v-u)^s}{2^{s-1}s!}.
\label{eq:inverse-modulus}
\end{equation}
Thus, for $0\le a<b\le 1$, that
\begin{equation*}
\int_a^b|f_j(t)|\,dt
\ge
\frac{(b-a)^{j+1}}{2^j(j+1)!}.
\end{equation*}
If $G'=|f_j|$ and $G$ is nondecreasing, then the minimum is attained at
$c=G((a+b)/2)$ and
\begin{align*}
\int_a^b|f_{j+1}(t)|\,dt
&\ge \inf_{c\in\mathbb R}\int_a^b|G(t)-c|\,dt \\
&=\int_0^{(b-a)/2}\int_{a+y}^{b-y}|f_j(t)|\,dt\,dy \\
&\ge \frac{1}{2^j(j+1)!}
\int_0^{(b-a)/2}(b-a-2y)^{j+1}\,dy \\
&=\frac{(b-a)^{j+2}}{2^{j+1}(j+2)!}.
\end{align*}
\subsection{One-observable localization}

\begin{theorem}
\label{thm:recursive-posterior-localization}
For every finite ensemble, every observable $0\preceq E\preceq I$, all integers $r,n\ge 1$,
and every $\eta\in(0,1]$, there is a deterministic decoder
$d_E:\mathcal X^r\to[0,1]$ such that the $r$-round posterior-refocused PGM satisfies
\begin{equation*}
\Pr\!\left[\left|d_E(H_r)-\Theta\right|>\eta\right]
\le
\left(\frac{9r^2}{n\eta^2}\right)^r.
\end{equation*}
The physical measurement is independent of $E$.
\end{theorem}

\begin{proof}
For each round $s$, introduce only for the analysis a fresh count
\[
K_s\mid X,H_{s-1}\sim\operatorname{Bin}(n,\Theta),
\qquad
\zeta_s:=K_s/n-\Theta.
\]
Let
\[
Q_s:=\left(\mathbb E\left|f_s^{H_s}(\Theta)\right|^2\right)^{1/2},
\qquad Q_0=1.
\]
Condition on $H_{s-1}=h$ and attach the labels
$z_x=F_{s-1}^{h}(\theta_x)$ to the posterior ensemble, the conditional expectation minimizes
mean-square error, and Theorem~\ref{thm:factor-two} may then be compared with the product measurement
$\{E,I-E\}^{\otimes n}$ reporting $F_{s-1}^{h}(K_s/n)$.  Averaging over the history gives
\begin{align*}
Q_s^2
&=\mathbb E\,\operatorname{Var}\!\left(
F_{s-1}^{H_{s-1}}(\Theta)\mid H_s\right) \notag\\
&\le \mathbb E\left|F_{s-1}^{H_{s-1}}(\theta_X)
-F_{s-1}^{H_{s-1}}(\theta_{Y_s})\right|^2 \notag\\
&\le 2\,\mathbb E\left|F_{s-1}^{H_{s-1}}(\Theta)
-F_{s-1}^{H_{s-1}}(K_s/n)\right|^2.
\end{align*}
Applying Eq.~\eqref{eq:triangular-increment}, Minkowski's inequality, and the binomial moment
bound from Appendix~A yields
\begin{align*}
Q_s
&\le \sqrt{2}\sum_{\ell=1}^{s}\frac{1}{\ell!}
\left\||\zeta_s|^\ell f_{s-\ell}^{H_{s-\ell}}(\Theta)\right\|_2 \notag\\
&\le \sqrt{2}\sum_{\ell=1}^{s}
\frac{\sqrt{b_{2\ell}}}{\ell!}\,n^{-\ell/2}Q_{s-\ell}.
\end{align*}
The second inequality uses fresh-block conditional independence given $X$ and the tower property
for the older features.  The definition in Appendix~A gives $b_2=1/4$ and
$b_{2\ell}=2^{1-\ell}\ell!\le \ell!$ for $\ell\ge2$.  Hence
\begin{equation}
Q_s\le \sqrt{2}\sum_{\ell=1}^{s}n^{-\ell/2}Q_{s-\ell}.
\label{eq:coarse-l2-recursion}
\end{equation}
An induction now gives
\begin{equation*}
Q_s\le \left(\frac{3}{\sqrt n}\right)^s.
\end{equation*}
Assuming the claim for all round smaller than $s$, Eq.~\eqref{eq:coarse-l2-recursion} implies
\[
Q_s
\le \sqrt{2}\,n^{-s/2}\sum_{\ell=1}^{s}3^{s-\ell}
=\frac{\sqrt{2}}{2}(3^s-1)n^{-s/2}
\le \left(\frac{3}{\sqrt n}\right)^s.
\]

\noindent
Fix a terminal history $h$ and define
\[
\mathcal G_h
:=\left\{x\in\mathcal X_h:\left|f_r^h(\theta_x)\right|<\frac{\eta^r}{r!}\right\}.
\]
If $x,x'\in\mathcal G_h$, then Eq.~\eqref{eq:inverse-modulus} gives
\[
\frac{|\theta_x-\theta_{x'}|^r}{2^{r-1}r!}
\le |f_r^h(\theta_x)-f_r^h(\theta_{x'})|
<\frac{2\eta^r}{r!},
\]
so $|\theta_x-\theta_{x'}|<2\eta$.  For $\mathcal G_h\neq\varnothing$, set
\[
d_E(h):=\frac12\left(\min_{x\in\mathcal G_h}\theta_x+\max_{x\in\mathcal G_h}\theta_x\right),
\]
and set $d_E(h)=0$ when $\mathcal G_h=\varnothing$.
Then
\begin{align*}
\Pr\!\left[|d_E(H_r)-\Theta|>\eta\right]
&\le \Pr\!\left[\left|f_r^{H_r}(\Theta)\right|
\ge \frac{\eta^r}{r!}\right] \\
&\le \frac{(r!)^2}{\eta^{2r}}Q_r^2
\le \frac{(r!)^2}{\eta^{2r}}\left(\frac{9}{n}\right)^r
\le \left(\frac{9r^2}{n\eta^2}\right)^r,
\end{align*}
\end{proof}

\subsection{Iteration and simultaneous concentration}

\begin{theorem}
\label{thm:finite-prior-high-confidence-cubic}
Fix a finite ensemble, one observable $E$, an accuracy $\eta\in(0,1]$, and a failure target
$\tau\in(0,1)$.  Set
\begin{equation}
r:=\left\lceil\log_2\frac{1}{\tau}\right\rceil,
\qquad
n:=\left\lceil\frac{18r^2}{\eta^2}\right\rceil.
\label{eq:one-coordinate-parameters}
\end{equation}
Then the $r$-round posterior-refocused PGM, followed by the decoder from
Theorem~\ref{thm:recursive-posterior-localization}, satisfies
\begin{equation*}
\Pr\!\left[\left|d_E(H_r)-\operatorname{Tr}(E\rho_X)\right|>\eta\right]
\le \tau.
\end{equation*}
Moreover,
\begin{equation*}
T=rn
=O\!\left(\frac{(1+\log(1/\tau))^3}{\eta^2}\right).
\end{equation*}
\end{theorem}

\begin{proof}
By Theorem~\ref{thm:recursive-posterior-localization},
\[
\Pr\!\left[|d_E(H_r)-\Theta|>\eta\right]
\le \left(\frac{9r^2}{n\eta^2}\right)^r
\le 2^{-r}
\le \tau.
\]
The copy bound follows from Eq.~\eqref{eq:one-coordinate-parameters}.
\end{proof}

\begingroup
\emergencystretch=2em
\begin{corollary}
\label{cor:simultaneous-finite-prior-cubic}
Fix a finite ensemble, observables $E_1,\ldots,E_M$, an accuracy $\epsilon\in(0,1]$, and
$\delta\in(0,1)$.  Set
\begin{equation*}
r:=\left\lceil\log_2\frac{M}{\delta}\right\rceil,
\qquad
n:=\left\lceil\frac{18r^2}{\epsilon^2}\right\rceil.
\label{eq:simultaneous-parameters}
\end{equation*}
There is a deterministic full-history decoder
$D:\mathcal X^r\to[0,1]^M$ such that the same $r$-round posterior-refocused PGM satisfies
\begin{equation*}
\Pr\!\left[
\max_{1\le j\le M}
\left|D_j(H_r)-\operatorname{Tr}(E_j\rho_X)\right|>\epsilon
\right]
\le \delta.
\label{eq:simultaneous-finite-prior}
\end{equation*}
The total number of copies obeys
\begin{equation*}
T=rn
=O\!\left(\frac{(1+\log(M/\delta))^3}{\epsilon^2}\right),
\end{equation*}
and is independent of the Hilbert-space dimension.
\end{corollary}
\endgroup

\begin{proof}
For each $E_j$, construct the decoder $d_{E_j}$ from
Theorem~\ref{thm:recursive-posterior-localization} and set
\[
D(h):=\bigl(d_{E_1}(h),\ldots,d_{E_M}(h)\bigr).
\]
The functions used to define the coordinates and decoders may depend on $E_j$, but the physical
posterior PGM depends only on the state ensemble and the observed history.  Hence the same history
$H_r$ is used for every coordinate. Thus, for each $j$,
\[
\Pr\!\left[
|D_j(H_r)-\operatorname{Tr}(E_j\rho_X)|>\epsilon
\right]
\le 2^{-r}
\le \frac{\delta}{M}.
\]
\end{proof}

%% file: sections/recovery_analysis.tex
\section{Finite prior analysis of the recovery measurement}
\label{sec:recovery-analysis}

We now consider Algorithm~\ref{alg:recovery-pgm} for the same finite ensemble.
The prior is fixed throughout this section.  Its output is the label $Y$, and
we use $\vartheta(\rho_Y)$ to represent decoder, and we first compare the selected rewards
in Lemma~\ref{lem:selected-reward}, then use this comparison to bound the
conditional mean bias in Corollary~\ref{cor:conditional-bias}.

\subsection{The output kernel and its symmetry}

For $0\le t<N$, define
\begin{equation}
K_t(y\mid x)
:=\operatorname{Tr}\!\left(D_{y,t}\rho_x^{\otimes t}\right),
\qquad
K(y\mid x):=\frac1N\sum_{t=0}^{N-1}K_t(y\mid x).
\label{eq:recovery-kernel}
\end{equation}
Thus $K_t(y\mid x)$ is the probability of outcome $y$ in the
$t$-copy measurement, conditional on $X=x$, by
Eq.~\eqref{eq:averaged-recovery-povm}, $K(y\mid x)$ is the probability
of $Y=y$ in Algorithm~\ref{alg:recovery-pgm}, conditional on $X=x$, and zero prior weights can be omitted when constructing this measurement.

\begin{lemma}
\label{lem:recovery-normalization}
The effects in Eq.~\eqref{eq:recovery-effects} form a POVM on
$\mathcal H^{\otimes t}$.  For every positive prior weight $p_x$,
$\supp(\rho_x^{\otimes t})\subseteq\supp(B_t)$.  Moreover,
\begin{equation}
\sum_yK_t(y\mid x)=1,\qquad
\sum_xp_xK_t(y\mid x)=p_y,\qquad
p_xK_t(y\mid x)=p_yK_t(x\mid y).
\label{eq:recovery-balance}
\end{equation}
The same identities hold with $K$ in place of $K_t$.
\end{lemma}

\begin{proof}
The inequality $B_t\succeq p_x\rho_x^{\otimes t}$ implies the support
inclusion.  If $v\in\ker B_t$, then
\[
0\le p_x\langle v,\rho_x^{\otimes t}v\rangle
\le\langle v,B_tv\rangle=0.
\]
Obviously, this implies $\rho_x^{\otimes t}v=0$, in particular, the
completion term in Eq.~\eqref{eq:recovery-effects} has zero probability
on every positive-weight ensemble state, it's not hard to verify that
the density $\beta_0$ has integral one,
so each integrand in Eq.~\eqref{eq:recovery-effects} is positive.  On
$\supp B_t$, summing the effects gives
\begin{align*}
\sum_y p_y\int_{\mathbb R}\beta_0(u)
B_t^{-(1+iu)/2}\rho_y^{\otimes t}B_t^{-(1-iu)/2}\,du
&=\int_{\mathbb R}\beta_0(u)
B_t^{-(1+iu)/2}B_tB_t^{-(1-iu)/2}\,du\\
&=P_t.
\end{align*}
Adding $\sum_yp_y(I-P_t)=I-P_t$ proves normalization on the full
space.

For the remaining identities, put $A_u=B_t^{-(1+iu)/2}$ on the
support. Thus,
\begin{align*}
p_xK_t(y\mid x)
&=p_xp_y\int_{\mathbb R}\beta_0(u)
\operatorname{Tr}\!\left(\rho_x^{\otimes t}A_u
\rho_y^{\otimes t}A_u^\dagger\right)\,du\notag\\
&=p_xp_y\int_{\mathbb R}\beta_0(u)
\operatorname{Tr}\!\left(\rho_y^{\otimes t}A_u^\dagger
\rho_x^{\otimes t}A_u\right)\,du\notag\\
&=p_yK_t(x\mid y).
\end{align*}
Here $\beta_0(-u)=\beta_0(u)$ and $A_{-u}=A_u^\dagger$.
Also,
\begin{align*}
\sum_xp_xK_t(y\mid x)
&=\operatorname{Tr}(B_tD_{y,t})\\
&=p_y\int_{\mathbb R}\beta_0(u)
\operatorname{Tr}\!\left(\rho_y^{\otimes t}P_t\right)\,du=p_y.
\end{align*}
At $t=0$, $D_{y,0}=p_y$, so the identities hold directly,
and averaging over $t$ proves the assertions for $K$. 
\end{proof}

\subsection{Recovery of a classical label}

We use natural logarithms in the entropy expressions.  We use von Neumann entropy, write
$S(\omega)=-\operatorname{Tr}(\omega\log\omega)$, and 
\[
D(\omega\Vert\tau)
=\operatorname{Tr}\!\left(\omega(\log\omega-\log\tau)\right)
\]
when $\supp\omega\subseteq\supp\tau$.  The root fidelity is
\[
F(\omega,\tau)=\|\sqrt\omega\sqrt\tau\|_1.
\]
For a classical law $q$, write $H(q)=-\sum_aq_a\log q_a$.
$H(J)_\omega$ denotes the entropy of the marginal law of a classical
register $J$ for subsystems of a state $\omega$, write
\[
I(U:V)_\omega:=S(\omega_U)+S(\omega_V)-S(\omega_{UV}),
\]
\[
I(J:C\mid K)_\omega
:=S(\omega_{JK})+S(\omega_{KC})-S(\omega_K)-S(\omega_{JKC}).
\]
The entropy identities used below are proved in
Lemma~\ref{lem:classical-entropy-budget}.
Strong subadditivity of Lieb and Ruskai gives
$I(J:C\mid K)_\omega\ge0$~\cite{liebruskai1973}.
Fawzi and Renner give a fidelity bound for recovery in terms of this
quantity~\cite{fawzirenner2015}.
We use the trace distance $d_{\rm tr}$ defined in
Section~\ref{sec:worst-case}.
The following recovery theorem is the external input to this section.

\begin{theorem}
\label{thm:universal-recovery}
Let $\omega,\tau$ be density matrices on a finite-dimensional space,
with $\supp\omega\subseteq\supp\tau$, and let $\mathcal N$ be a
quantum channel, set $B=\mathcal N(\tau)$ and let $\mathcal N^*$ be
its trace adjoint.  On $\supp B$, define
\begin{equation}
\mathcal R^{[u]}_{\tau,\mathcal N}(X)
=\tau^{(1-iu)/2}\mathcal N^*\!\left(
B^{-(1-iu)/2}XB^{-(1+iu)/2}\right)\tau^{(1+iu)/2},
\label{eq:rotated-recovery-map}
\end{equation}
\begin{equation*}
\mathcal R_{\tau,\mathcal N}(X)
=\int_{\mathbb R}\beta_0(u)\mathcal R^{[u]}_{\tau,\mathcal N}(X)\,du.
\end{equation*}
These maps are trace preserving on the supported input space and can be
completed to channels on the full space.  The averaged map satisfies
\begin{equation*}
D(\omega\Vert\tau)-D(\mathcal N(\omega)\Vert\mathcal N(\tau))
\ge-2\log F\!\left(\omega,
(\mathcal R_{\tau,\mathcal N}\circ\mathcal N)(\omega)\right).
\end{equation*}
\end{theorem}

This is the finite-dimensional specialization of
\cite[Theorem~2.1 and Remark~2.2]{junge2018recovery}.
All inverse powers are restricted to supports, and positive powers
with positive real part are zero on the kernel, and the full-space
completion used for classical labels is proved in
Appendix~\ref{app:recovery-specialization}.

For an integer $k\ge1$, write $[k]=\{1,\ldots,k\}$.  Fix a stochastic
map $W:\mathcal X\to[k]$, so that $W(a\mid x)\ge0$ and
$\sum_aW(a\mid x)=1$ for every $x$.  Its label is only used
for the analysis and not for the physical measurement. 
 Let $J$ be the classical register storing this label,
and let $Q_1,\ldots,Q_N$ denote quantum registers with
Hilbert space $\mathcal H$, and write $Q^t=Q_1\cdots Q_t$ for their
first $t$ registers, with $Q^0$ the trivial register.
Define the following state only for the analysis.
\begin{equation*}
\Omega_{JQ^N}
:=\sum_{x,a}p_xW(a\mid x)\ketbra{a}{a}\otimes\rho_x^{\otimes N}.
\end{equation*}
For $0\le t<N$, its marginal satisfies $\Omega_{Q^t}=B_t$, and
the $a$-block of $\Omega_{JQ^t}$ is the unnormalized matrix
\[
\sum_xp_xW(a\mid x)\rho_x^{\otimes t}\preceq B_t.
\]
This inequality gives its support inclusion.  If
$\sum_xp_xW(a\mid x)=0$, the block is zero, and that label
contributes a zero effect and can be omitted.

Apply Eq.~\eqref{eq:rotated-recovery-map} to the reference
$\Omega_{JQ^t}$ and the channel $\operatorname{Tr}_J$.
Since its adjoint inserts the identity on $J$, the $a$-block of
one rotated recovery output is
\begin{equation*}
\begin{aligned}
&\left(\sum_xp_xW(a\mid x)\rho_x^{\otimes t}\right)^{(1-iu)/2}
 B_t^{-(1-iu)/2}X B_t^{-(1+iu)/2}
 \left(\sum_xp_xW(a\mid x)\rho_x^{\otimes t}\right)^{(1+iu)/2}.
\end{aligned}
\end{equation*}
Taking the trace combines the two powers of
$\sum_xp_xW(a\mid x)\rho_x^{\otimes t}$ into that same sum, 
on the complement of $P_t$, complete the recovered label with
$\left(\sum_xp_xW(a\mid x)\right)_{a\in[k]}$.
Its full space classical effect for label $a$ is therefore
\begin{equation}
\int_{\mathbb R}\beta_0(u)
 B_t^{-(1+iu)/2}\left(\sum_xp_xW(a\mid x)\rho_x^{\otimes t}\right)
 B_t^{-(1-iu)/2}\,du+\left(\sum_xp_xW(a\mid x)\right)(I-P_t)
 =\sum_xW(a\mid x)D_{x,t}.
\label{eq:coarse-recovery-consistency}
\end{equation}
This identity concerns the recovered classical output after discarding
the recovered $Q^t$.  In particular, every choice of $W$ gives a
classical post-processing of the same effects $D_{x,t}$.  We use this
identity in Eq.~\eqref{eq:selected-recovered-state} below.
\subsection{Selected rewards and conditional mean bias}

\begin{lemma}
\label{lem:selected-reward}
For every stochastic map $W:\mathcal X\to[k]$ and every list of
effects $0\preceq A_1,\ldots,A_k\preceq I$ on $\mathcal H$, the
kernel in Eq.~\eqref{eq:recovery-kernel} satisfies
\begin{equation}
\left|
\sum_{x,a}p_xW(a\mid x)\operatorname{Tr}(A_a\rho_x)
-\sum_{x,y,a}p_xK(y\mid x)W(a\mid y)
\operatorname{Tr}(A_a\rho_x)
\right|
\le\sqrt{\frac{\log k}{N}}.
\label{eq:selected-reward-bound}
\end{equation}
The same kernel $K$ works for every $W$ and every such list of effects.
\end{lemma}

\begin{proof}
Set
\[
\eta_t:=I(J:Q_{t+1}\mid Q^t)_\Omega,
\qquad 0\le t<N.
\]
By the entropy identities from Lemma~\ref{lem:classical-entropy-budget},
\begin{equation*}
\eta_t\ge0,\qquad
\sum_{t=0}^{N-1}\eta_t
=I(J:Q^N)_\Omega\le H(J)_\Omega\le\log k.
\end{equation*}
For a fixed $t$, write $C=Q_{t+1}$.  Apply
Theorem~\ref{thm:universal-recovery} with input $\Omega_{JQ^tC}$,
reference $\Omega_{JQ^t}\otimes\Omega_C$, and channel
$\operatorname{Tr}_J$.  Lemma~\ref{lem:support-distance-fidelity}
gives the required support inclusion.  The relative-entropy loss is
\begin{align*}
&D(\Omega_{JQ^tC}\Vert\Omega_{JQ^t}\otimes\Omega_C)
-D(\Omega_{Q^tC}\Vert B_t\otimes\Omega_C)\notag\\
&\qquad=S(\Omega_{JQ^t})+S(\Omega_{Q^tC})
-S(B_t)-S(\Omega_{JQ^tC})=\eta_t.
\end{align*}
In Eq.~\eqref{eq:rotated-recovery-map}, the powers of $\Omega_C$
cancel on its support, so the recovery acts as the map associated
with $\Omega_{JQ^t}$ on $Q^t$, tensor the identity on $C$.
Appendix~\ref{app:recovery-specialization} proves this cancellation
also when $\Omega_C$ is singular.  It agrees with the stabilization
property in \cite{junge2018recovery} after restricting
the spectator register to $\supp\Omega_C$.

Let $\widetilde\Omega_{JQ^tC}$ be the recovered state.  Theorem~\ref{thm:universal-recovery}
and Lemma~\ref{lem:support-distance-fidelity} give, using
\cite[Theorem 1]{fuchsvandegraaf1999} for the trace-distance bound,
\begin{equation*}
F(\Omega_{JQ^tC},\widetilde\Omega_{JQ^tC})\ge e^{-\eta_t/2},
\qquad
 d_{\rm tr}(\Omega_{JQ^tC},\widetilde\Omega_{JQ^tC})
\le\sqrt{1-e^{-\eta_t}}\le\sqrt{\eta_t}.
\end{equation*}
After discarding the recovered $Q^t$, Eq.~\eqref{eq:coarse-recovery-consistency}
gives
\begin{equation}
\widetilde\Omega_{JC}
=\sum_{x,y,a}p_xK_t(y\mid x)W(a\mid y)
\ketbra{a}{a}\otimes\rho_x.
\label{eq:selected-recovered-state}
\end{equation}
Conditional on the original label $X=x$, the prefix has state
$\rho_x^{\otimes t}$ and the untouched copy $C$ has state $\rho_x$.
The probability of fine outcome $y$ is $K_t(y\mid x)$, followed by
classical transition $W(a\mid y)$.  Conditional independence has only
been used at $X=x$, thus, given a coarse label $a$, the conditional state
can be a mixture of tensor powers.

\noindent
Since the operator
\[
\sum_a\ketbra{a}{a}\otimes I_{Q^t}\otimes A_a
\]
is an effect, its expectation on the original state is the first
term of Eq.~\eqref{eq:selected-reward-bound}, and its expectation on
the recovered state is the second term with $K_t$ in place of $K$.
Thus the difference between the two expectations, denoted by $\Delta_t$, satisfies
$|\Delta_t|\le\sqrt{\eta_t}$.  Averaging over every $t$ gives
\begin{equation*}
\left|\frac1N\sum_{t=0}^{N-1}\Delta_t\right|
\le\frac1N\sum_{t=0}^{N-1}\sqrt{\eta_t}
\le\sqrt{\frac1N\sum_{t=0}^{N-1}\eta_t}
\le\sqrt{\frac{\log k}{N}}.
\end{equation*}
\end{proof}

\begin{corollary}
\label{cor:conditional-bias}
For the known effects $E_1,\ldots,E_M$, Algorithm~\ref{alg:recovery-pgm}
satisfies
\begin{equation}
\sum_xp_x\max_{j\in[M]}
\left|\theta_j(\rho_x)-\sum_yK(y\mid x)\theta_j(\rho_y)\right|
\le\sqrt{\frac{\log(2M)}{N}}.
\label{eq:conditional-bias-bound}
\end{equation}
\end{corollary}

\begin{proof}
Take $k=2M$, with labels $(j,s)\in[M]\times\{-1,+1\}$, and set
\[
A_{j,+}=E_j,\qquad A_{j,-}=I-E_j.
\]
For any stochastic $W$, Eq.~\eqref{eq:recovery-balance} gives the
following change of indices in the second reward.
\begin{align*}
&\sum_{x,y,a}p_xK(y\mid x)W(a\mid y)\operatorname{Tr}(A_a\rho_x)
\notag\\
&\quad=\sum_{x,y,a}p_yK(x\mid y)W(a\mid y)\operatorname{Tr}(A_a\rho_x)
\notag\\
&\quad=\sum_{x,y,a}p_xK(y\mid x)W(a\mid x)\operatorname{Tr}(A_a\rho_y).
\end{align*}
Since $\sum_yK(y\mid x)=1$, the difference of rewards is therefore
\begin{equation}
\sum_{x,j,s}p_xW(j,s\mid x)s
\left(\theta_j(\rho_x)-\sum_yK(y\mid x)\theta_j(\rho_y)\right).
\label{eq:signed-reward-cancellation}
\end{equation}
For $s=-1$, the two constant terms from $I-E_j$ cancel in this
expression.  For each $x$, choose a pair $(j,s)$ maximizing the
signed difference, with deterministic tie breaking, and let
$W(\cdot\mid x)$ be the point mass at that pair, this $W$ is only used for helping with analysis, the kernel $K$ has already been constructed from the ensemble and thus is independent of $W$, so Lemma~\ref{lem:selected-reward} can be applied to this choice as well. Then
Eq.~\eqref{eq:signed-reward-cancellation} is the left side of
Eq.~\eqref{eq:conditional-bias-bound}. 
\end{proof}

%% file: sections/worst_case.tex
\section{Worst-case shadow tomography}
\label{sec:worst-case}

Here $C$ denotes an absolute positive constant, which may increase between
inequalities.  Every measurement budget uses a fixed sufficiently large
choice of this constant before priors or histories are considered.
\subsection{A common finite decoder}
\label{subsec:common-finite-decoder}

We design a shadow tomography strategy into 2 parts, a measurement and a finite decoder.
The finite-prior construction from
Corollary~\ref{cor:simultaneous-finite-prior-cubic} uses the complete
history $H_r$ and its decoder $D(H_r)$.  Algorithm~\ref{alg:recovery-pgm}
uses the label $Y$ and its decoder $\vartheta(\rho_Y)$.
For a generic finite-output POVM $\{Q_v\}_{v\in\mathcal V}$ with decoder
$d:\mathcal V\to[0,1]^M$, we will absorb the rounding into its effects
by summing all $Q_v$ with the same rounded value.  This places both
constructions in the space used in Lemma~\ref{lem:finite-minimax}.
For the minimax argument, we need to represent the strategy using same finite outcome alphabet and the same
decoder.  We achieve this by rounding each possible expectation vector to a
fixed grid on a designated net. And we claim that after this rounding step, we do not lose too much precision, while placing all strategies in one common finite dimensional space.

Fix an accuracy \(\epsilon\in(0,1]\), and define
\[
  K_\epsilon
  \coloneqq
  \left\lceil\frac{2}{\epsilon}\right\rceil,
  \qquad
  \mathcal Z_\epsilon
  \coloneqq
  \left\{
    \frac{k}{K_\epsilon}:0\leq k\leq K_\epsilon
  \right\}^{M}.
\]
Let
\[
  \mathcal R_\epsilon
  \colon
  [0,1]^M\longrightarrow\mathcal Z_\epsilon
\]
be coordinatewise rounding, with any fixed rule
for breaking ties.  Since the grid spacing is \(1/K_\epsilon\),
\begin{equation}
  \left\|
    \mathcal R_\epsilon(v)-v
  \right\|_\infty
  \leq
  \frac{1}{2K_\epsilon}
  \leq
  \frac{\epsilon}{4}
  \qquad
  \text{for every }v\in[0,1]^M.
  \label{eq:grid-rounding-error}
\end{equation}

Now consider a finite ensemble
\[
  \{(p_x,\rho_x):x\in\mathcal X\},
\]
and let \(H_r\) be the complete history produced by its sequential PGM.  We
replace the original decoder by
\[
  Z
  \coloneqq
  \mathcal R_\epsilon
  \bigl(
    D(H_r)
  \bigr)
  \in\mathcal Z_\epsilon.
\]
The rounded output remains accurate whenever the original full-history decoder is
accurate at half the target tolerance.  Indeed, if
\[
  \left\|
    D(H_r)-\vartheta(\rho_X)
  \right\|_\infty
  \leq
  \frac{\epsilon}{2},
\]
then
\[
\begin{aligned}
  \left\|
    Z-\vartheta(\rho_X)
  \right\|_\infty
  &\leq
  \left\|
    Z-D(H_r)
  \right\|_\infty
  +
  \left\|
    D(H_r)-\vartheta(\rho_X)
  \right\|_\infty                                                    \\
  &\leq
  \frac{\epsilon}{4}
  +
  \frac{\epsilon}{2}
  <
  \epsilon.
\end{aligned}
\]
So,
\begin{equation*}
  \left\{
    \left\|
      Z-\vartheta(\rho_X)
    \right\|_\infty
    >
    \epsilon
  \right\}
  \subseteq
  \left\{
    \left\|
      D(H_r)-\vartheta(\rho_X)
    \right\|_\infty
    >
    \frac{\epsilon}{2}
  \right\}.
  \end{equation*}

Which means the rounding operation can be absorbed into the measurement step.
In the generic notation above, set
\[
  N_z\coloneqq\sum_{\substack{v\in\mathcal V:\\
  \mathcal R_\epsilon(d(v))=z}}Q_v.
\]
The operators are positive, and the sets in this sum partition
$\mathcal V$, so $\sum_zN_z=\sum_vQ_v=I$.  Empty sums are zero.
In particular, for Algorithm~\ref{alg:sequential-pgm} with prior $p$,
\begin{equation}
  N_z^{(p)}
  \coloneqq
  \sum_{\substack{h\in\mathcal X^r:\\
    \mathcal R_\epsilon(D^{(p)}(h))=z}}
  \widetilde M_h^{[r],(p)}.
  \label{eq:coarse-grained-grid-povm}
\end{equation}
Here $D^{(p)}$ is the decoder in
Corollary~\ref{cor:simultaneous-finite-prior-cubic}. The superscript only
makes its prior dependence explicit.  Thus
$Z=\mathcal R_\epsilon(D^{(p)}(H_r))$.  The full history effects are
Eq.\eqref{eq:full-history-povm}.
For Algorithm~\ref{alg:recovery-pgm}, the same construction is
\[
  N_z^{(p)}
  =\sum_{\substack{y\in\mathcal X:\\
  \mathcal R_\epsilon(\vartheta(\rho_y))=z}}\overline D_y^{(p)}.
\]
In either case, the decoder is the identity map
\[
  g(z)=z.
\]

The POVM operators in
Eq.\eqref{eq:coarse-grained-grid-povm} may still depend on the prior, as is
appropriate for the finite prior problem.  The outcome alphabet
\(\mathcal Z_\epsilon\) and the decoder \(g(z)=z\) depend only on
\(M\) and \(\epsilon\).  Thus every prior dependent Bayes strategy is now
represented as an element of the same finite dimensional POVM space.

\subsection{Finite-state minimax}
\label{subsec:finite-state-minimax}

We now regard shadow tomography on a fixed finite family as a zero-sum game.
The adversary chooses a loss index, the measurement player chooses a
grid-output POVM, and the payoff is that loss.  The
finite prior theorems control the value of the two games below against every
distribution over their loss indices.  Minimax will convert these prior-dependent Bayes
guarantees into one measurement that controls every loss index in the family.
Fix a positive integer \(T\).  Let
\[
  \mathcal L_T(\mathcal Z_\epsilon)
  \coloneqq
  \left\{
    (N_z)_{z\in\mathcal Z_\epsilon}:
    N_z\succeq0,\quad
    \sum_{z\in\mathcal Z_\epsilon}N_z=I
  \right\}
\]
be the set of all \(\mathcal Z_\epsilon\)-outcome POVMs on
\(\mathcal H^{\otimes T}\).

This set is nonempty (put the identity at one grid point and zero elsewhere).
This set is convex and also compact, since it is a closed and bounded subset of the
finite dimensional real vector space of Hermitian
\(\mathcal Z_\epsilon\)-indexed tuples, and the normalization condition
implies
\[
  0\preceq N_z\preceq I
\]
for every outcome \(z\), so the set is closed,
for a state \(\rho\in\mathcal D(\mathcal H)\) and a tolerance \(a>0\), let
\[
  B_a(\rho)
  \coloneqq
  \left\{
    z\in\mathcal Z_\epsilon:
    \left\|
      z-\vartheta(\rho)
    \right\|_\infty
    >
    a
  \right\}
\]
be the set of outputs that fail at \(\rho\).  For
\(N\in\mathcal L_T(\mathcal Z_\epsilon)\), define
\begin{equation*}
  f_{\rho,a}(N)
  \coloneqq
  \sum_{z\in B_a(\rho)}
  \operatorname{Tr}
  \bigl(
    N_z\rho^{\otimes T}
  \bigr).
  \end{equation*}
Thus \(f_{\rho,a}(N)\) is the failure probability of the
identity-decoded strategy \(N\) on input \(\rho\), thus for a fixed \(\rho\) and
\(a\), it is a continuous affine function of \(N\), 
for event \(A\subseteq\mathcal Z_\epsilon\), we write
\begin{equation*}
  \mathbb P_\rho^N[A]
  \coloneqq
  \sum_{z\in A}
  \operatorname{Tr}
  \bigl(
    N_z\rho^{\otimes T}
  \bigr).
  \end{equation*}

\begin{lemma}
\label{lem:finite-minimax}
Let $\mathcal L$ be a nonempty compact convex subset of a
finite dimensional real vector space, and let
$f_i:\mathcal L\to\mathbb R$, $i\in\mathcal I$, be finitely many
continuous affine functions, with $\mathcal I$ nonempty.  Then
\begin{equation}
  \min_{N\in\mathcal L}\max_{i\in\mathcal I}f_i(N)
  =
  \max_{\pi\in\mathcal P(\mathcal I)}\min_{N\in\mathcal L}
  \sum_{i\in\mathcal I}\pi_i f_i(N).
  \label{eq:finite-minimax-identity}
\end{equation}
Both extrema are attained.  In particular, if each mixed adversary has
a response with expected loss at most $q$, one $N$ has every loss at
most $q$.
\end{lemma}

\begin{proof}
For $N\in\mathcal L$ and $\pi\in\mathcal P(\mathcal I)$, define
\[
  J(N,\pi)\coloneqq\sum_{i\in\mathcal I}\pi_i f_i(N).
\]
The function $J$ is continuous and affine in each argument.  Both
$\mathcal L$ and $\mathcal P(\mathcal I)$ are compact and convex,
so Sion's minimax theorem~\cite[Theorem~3.4, p.~174]{sion1958} gives
\[
  \min_{N\in\mathcal L}\max_{\pi\in\mathcal P(\mathcal I)}J(N,\pi)
  =\max_{\pi\in\mathcal P(\mathcal I)}\min_{N\in\mathcal L}J(N,\pi).
\]
For a fixed measurement $N$,
\[
  \max_{\pi\in\mathcal P(\mathcal I)}J(N,\pi)
  =\max_{i\in\mathcal I}f_i(N),
\]
because the maximum over the probability simplex is achieved at a point
mass.  Substituting this identity proves
Eq.\eqref{eq:finite-minimax-identity}.
The minimum on the left is attained by continuity on a compact set.
For the other extremum, compactness makes
$C_0=\max_{i,N}|f_i(N)|$ finite, and
\[
  \left|\min_NJ(N,\pi)-\min_NJ(N,\pi')\right|
  \leq C_0\sum_i|\pi_i-\pi_i'|.
\]
Thus that minimum is continuous in $\pi$, and its maximum is attained.
\end{proof}

The minimization in Lemma~\ref{lem:finite-minimax} is over the full convex
set of grid-output POVMs.  We do not require the prior-dependent sequential
PGMs themselves to form a convex family.  Their role is only to provide 
one feasible measurement with small Bayes failure probability.  The same
full POVM space is used for the signed losses below, with its own fixed
budget.  No implementation structure of either prior-dependent algorithm
is imposed on the minimax optimizer.

\begin{proposition}
\label{prop:finite-uniformization}
Let
\(\mathcal F\subseteq\mathcal D(\mathcal H)\) be finite and nonempty, let
\(\epsilon\in(0,1]\), and let \(\beta\in(0,1)\).  Let \(T\) be the copy
count supplied by
Corollary~\ref{cor:simultaneous-finite-prior-cubic} when it is applied with
accuracy \(\epsilon/2\) and total failure probability \(\beta\).  Then there
exists a POVM
\[
  N^{\mathcal F}
  \in
  \mathcal L_T(\mathcal Z_\epsilon)
\]
such that
\begin{equation*}
  \max_{\rho\in\mathcal F}
  f_{\rho,\epsilon}
  \bigl(
    N^{\mathcal F}
  \bigr)
  \leq
  \beta.
  \label{eq:finite-state-uniformization}
\end{equation*}
In particular, if
\[
  L_\beta
  \coloneqq
  \log\frac{2M}{\beta},
\]
then
\begin{equation*}
  T
  \leq
  \frac{CL_\beta^3}{\epsilon^2}
  \label{eq:finite-state-copy-bound}
\end{equation*}
for a universal constant \(C\).
\end{proposition}

\begin{proof}
Fix an arbitrary prior
\(
  p\in\mathcal P(\mathcal F)
\). Applying Corollary~\ref{cor:simultaneous-finite-prior-cubic} with accuracy
\(\epsilon/2\) and failure probability \(\beta\) gives a sequential PGM
whose full-history decoder \(D(H_r)\) satisfies
\[
  \mathbb P
  \left[
    \left\|
      D(H_r)-\vartheta(\rho_X)
    \right\|_\infty
    >
    \frac{\epsilon}{2}
  \right]
  \leq
  \beta.
\]

Complete all zero-probability histories arbitrarily, flatten the sequential
strategy into a \(T\)-copy POVM, and apply the grid post-processing from the
preceding subsection.  This produces a measurement
\[
  N^{(p)}
  \in
  \mathcal L_T(\mathcal Z_\epsilon)
\]
on the common outcome alphabet.  The rounding guarantee implies
\[
\begin{aligned}
  \sum_{\rho\in\mathcal F}
  p_\rho
  f_{\rho,\epsilon}
  \bigl(
    N^{(p)}
  \bigr)
  &=
  \mathbb P
  \left[
    \left\|
      Z-\vartheta(\rho_X)
    \right\|_\infty
    >
    \epsilon
  \right]                                                        \\
  &\leq
  \mathbb P
  \left[
    \left\|
      D(H_r)-\vartheta(\rho_X)
    \right\|_\infty
    >
    \frac{\epsilon}{2}
  \right]                                                        \\
  &\leq
  \beta.
\end{aligned}
\]
Therefore, for every prior \(p\),
\[
  \min_{N\in\mathcal L_T(\mathcal Z_\epsilon)}
  \sum_{\rho\in\mathcal F}
  p_\rho f_{\rho,\epsilon}(N)
  \leq
  \beta.
\]
Taking the maximum over \(p\) and applying
Lemma~\ref{lem:finite-minimax}, with losses indexed by the states, gives
\[
  \min_{N\in\mathcal L_T(\mathcal Z_\epsilon)}
  \max_{\rho\in\mathcal F}
  f_{\rho,\epsilon}(N)
  \leq
  \beta.
\]
Finally,
Corollary~\ref{cor:simultaneous-finite-prior-cubic} gives
\[
  T
  \in
  O\left(
    \frac{L_\beta^3}{\epsilon^2}
  \right),
\]
\end{proof}

\subsection{The trace-distance net}
\label{subsec:trace-distance-net}
The major usage of the trace-distance net is to lift our finite ensemble conclusion to all states. 
We use the trace distance
\[
  d_{\mathrm{tr}}(\rho,\sigma)
  =
  \frac12\|\rho-\sigma\|_1.
\]
Because \(\mathcal D(\mathcal H)\) is compact in finite dimension, for every
\(s>0\) there is a finite set \(\mathcal F_s\subset\mathcal D(\mathcal H)\)
such that every state \(\rho\) has some \(\sigma\in\mathcal F_s\) with
\(d_{\mathrm{tr}}(\rho,\sigma)\leq s\).

\begin{lemma}
  \label{lem:trace-distance-stability}
  If \(d_{\mathrm{tr}}(\rho,\sigma)\leq s\), then
  \begin{equation*}
    \|\vartheta(\rho)-\vartheta(\sigma)\|_\infty\leq s
    \label{eq:effect-trace-stability}
  \end{equation*}
  and, for every positive integer \(T\),
  \begin{equation*}
    d_{\mathrm{tr}}(\rho^{\otimes T},\sigma^{\otimes T})
    \leq Ts.
    \label{eq:tensor-trace-stability}
  \end{equation*}
  So, for every POVM \(N\) and every event \(A\) in its output
  alphabet,
  \begin{equation*}
    \left|\mathbb P_\rho^N[A]-\mathbb P_\sigma^N[A]\right|
    \leq Ts.
    \label{eq:measurement-law-stability}
  \end{equation*}
\end{lemma}

\begin{proof}
  Put \(D=\rho-\sigma\).  Since \(D\) is Hermitian and traceless, its
  Jordan decomposition \(D=D_+-D_-\) satisfies
  \[
    \operatorname{Tr}D_+
    =
    \operatorname{Tr}D_-
    =
    \frac12\|D\|_1.
  \]
  
  The telescoping identity
  \[
    \rho^{\otimes T}-\sigma^{\otimes T}
    =
    \sum_{t=1}^{T}
      \rho^{\otimes(t-1)}\otimes(\rho-\sigma)
      \otimes\sigma^{\otimes(T-t)}
  \]
  and multiplicativity of the trace norm under tensor products give
  \[
    \frac12\|\rho^{\otimes T}-\sigma^{\otimes T}\|_1
    \leq
    \frac{T}{2}\|\rho-\sigma\|_1,
  \]
Finally, for an output event
  \(A\), the operator \(Q_A=\sum_{z\in A}N_z\) is an observable.  Applying the
  first part of the proof to \(Q_A\) on the tensor-product space proves
  the lemma.
\end{proof}

\begin{theorem}
  \label{thm:finite-net-lift}
  Let \(0<\epsilon<\varepsilon\leq1\) and
  \(0<\beta<\delta<1\).  Fix \(T\), and choose
  \begin{equation*}
    0<s\leq
    \min\left\{
      \varepsilon-\epsilon,
      \frac{\delta-\beta}{T}
    \right\}.
    \label{eq:net-radius-general}
  \end{equation*}
  If a \(\mathcal Z_\epsilon\)-outcome POVM \(N\) satisfies
  \[
    \max_{\sigma\in\mathcal F_s}f_{\sigma,\epsilon}(N)\leq\beta
  \]
  on a finite trace-distance \(s\)-net, then
  \begin{equation}
    \sup_{\rho\in\mathcal D(\mathcal H)}
    f_{\rho,\varepsilon}(N)
    \leq\delta.
    \label{eq:finite-net-lift}
  \end{equation}
\end{theorem}

\begin{proof}
  Fix \(\rho\in\mathcal D(\mathcal H)\), and choose
  \(\sigma\in\mathcal F_s\) with
  \(d_{\mathrm{tr}}(\rho,\sigma)\leq s\).  Define the two subsets of the
  common output alphabet
  \[
    A_\rho
    =
    \{z:\|z-\vartheta(\rho)\|_\infty>\varepsilon\},
    \qquad
    A_\sigma
    =
    \{z:\|z-\vartheta(\sigma)\|_\infty>\epsilon\}.
  \]
  Lemma~\ref{lem:trace-distance-stability} implies \(A_\rho\subseteq A_\sigma\), if
  \(z\notin A_\sigma\), then
  \[
    \|z-\vartheta(\rho)\|_\infty
    \leq
    \|z-\vartheta(\sigma)\|_\infty
    +\|\vartheta(\sigma)-\vartheta(\rho)\|_\infty
    \leq\epsilon+s\leq\varepsilon.
  \]
  Using this inclusion and Lemma~\ref{lem:trace-distance-stability},
  \begin{align*}
    f_{\rho,\varepsilon}(N)
    &=\mathbb P_\rho^N[A_\rho]\\
    &\leq\mathbb P_\rho^N[A_\sigma]\\
    &\leq\mathbb P_\sigma^N[A_\sigma]+Ts\\
    &=f_{\sigma,\epsilon}(N)+Ts\\
    &\leq\beta+(\delta-\beta)
    =\delta.
  \end{align*}
  Since \(\rho\) was arbitrary, Eq.\eqref{eq:finite-net-lift} follows.
\end{proof}

The net is not a family over which we take a union bound.  Minimax first
produces one POVM that is uniform over all points of the net, and trace-distance
stability then extends that same POVM to the entire state space.  Therefore
the cardinality of \(\mathcal F_s\), and hence the Hilbert-space dimension 
does not enter the copy count.

\begin{theorem}
\label{thm:sequential-uniform}
Under the hypotheses of Theorem~\ref{thm:main}, there is a state-uniform
shadow-tomography strategy using
\[
  T=O\!\left(\frac{\log^3(2M/\delta)}{\varepsilon^2}\right)
\]
copies, obtained by uniformizing the full-history guarantee of
Algorithm~\ref{alg:sequential-pgm}.
\end{theorem}

\begin{proof}
Set $\epsilon=\varepsilon/2$, $\beta=\delta/2$, and fix the integers
$r,n$ from Corollary~\ref{cor:simultaneous-finite-prior-cubic} at accuracy
$\epsilon/2$ and failure $\beta$, so $T=rn$ is fixed before choosing
the net.  Take $s=\min\{\varepsilon/2,\delta/(2T)\}$.
Proposition~\ref{prop:finite-uniformization} supplies one measurement on
$\mathcal F_s$, and Theorem~\ref{thm:finite-net-lift} extends it to every
state. 
\end{proof}

\subsection{Minimax on recovery measurements and independent averaging}
\label{subsec:uniform-bias}

We now apply the same minimax lemma to a different algorithm, the conditional
mean bound in Corollary~\ref{cor:conditional-bias} concerns the mean of the
recovery label decoder.  We first make this mean bound uniform over all
states, and then use independent copies of the final measurement in
Corollary~\ref{cor:independent-averaging}.

\begin{lemma}
\label{lem:uniform-bias}
Fix $0\preceq E_1,\ldots,E_M\preceq I$ and $\epsilon\in(0,1)$.  For a sufficiently large absolute constant $C$, set
\begin{equation*}
N:=\left\lceil\frac{C\log(2M)}{\epsilon^2}\right\rceil.
\end{equation*}
There is a $\mathcal Z_\epsilon$-outcome POVM on $\mathcal H^{\otimes N}$, with
identity decoder, such that
\begin{equation}
\sup_{\rho\in\mathcal D(\mathcal H)}
\max_{1\le j\le M}
\left|\theta_j(\rho)-\mathbb E_\rho Z_j\right|
\le \epsilon,
\qquad \theta_j(\rho):=\operatorname{Tr}(E_j\rho).
\label{eq:uniform-bias}
\end{equation}
The measurement depends only on $\mathcal H$, the known list, and $\epsilon$.
\end{lemma}

\begin{proof}
First fix a finite nonempty state set
$\mathcal F=\{\rho_x:x\in\mathcal X\}$.  Use the common space
$\mathcal L_N(\mathcal Z_\epsilon)$, so both the copy count and the output alphabet
are fixed before choosing a prior.  The adversary now chooses
$(x,j,s)\in\mathcal X\times[M]\times\{-1,+1\}$, with payoff
\begin{equation*}
f_{x,j,s}((N_z)_z)
:=s\left(
\theta_j(\rho_x)-
\sum_{z\in\mathcal Z_\epsilon}z_j\operatorname{Tr}(N_z\rho_x^{\otimes N})
\right).
\end{equation*}
For each fixed $(x,j,s)$ this is a continuous affine function on the common
POVM space, as required in Lemma~\ref{lem:finite-minimax}.

Fix a distribution $\pi(x,j,s)$ over these choices and put
\[
p_x:=\sum_{j,s}\pi(x,j,s),
\qquad
W((j,s)\mid x):=\frac{\pi(x,j,s)}{p_x}\quad(p_x>0).
\]
States of weight zero can be omitted when constructing the recovery
measurement. Choose $W(\cdot\mid x)$ arbitrarily for them.  There are $2M$
signed-coordinate labels, and set
\[
A_{j,+}=E_j,\qquad A_{j,-}=I-E_j.
\]
Apply the recovery label measurement with the state prior $p$ and budget
$N$, and first decode its label $Y$ as $\vartheta(\rho_Y)$.  In particular,
the conditional law of $Y$ given $X=x$ is the same averaged kernel $K(y\mid x)$
for every $W$, as in Lemma~\ref{lem:selected-reward}.
By detailed balance, Eq.~\eqref{eq:recovery-balance}, the second reward in
that lemma can be written as
\begin{align*}
&\sum_{x,y,j,s}p_xK(y\mid x)W((j,s)\mid y)
\operatorname{Tr}(A_{j,s}\rho_x)\\
&\qquad=
\sum_{x,y,j,s}p_yK(x\mid y)W((j,s)\mid y)
\operatorname{Tr}(A_{j,s}\rho_x)\\
&\qquad=
\sum_{x,j,s}p_xW((j,s)\mid x)
\sum_yK(y\mid x)\operatorname{Tr}(A_{j,s}\rho_y).
\end{align*}
Here the last equality interchanges $x$ and $y$.  Since
$\sum_yK(y\mid x)=1$, both signs satisfy
\[
\operatorname{Tr}(A_{j,s}\rho_x)
-\sum_yK(y\mid x)\operatorname{Tr}(A_{j,s}\rho_y)
=s\left(\theta_j(\rho_x)-\sum_yK(y\mid x)\theta_j(\rho_y)\right).
\]
For $s=-1$, the two constant terms from $I-E_j$ cancel.  Therefore
Lemma~\ref{lem:selected-reward} gives
\begin{equation}
\sum_{x,j,s}\pi(x,j,s)s
\left(\theta_j(\rho_x)-\sum_yK(y\mid x)\theta_j(\rho_y)\right)
\le\sqrt{\frac{\log(2M)}{N}}\le\frac{\epsilon}{2}.
\label{eq:mixed-signed-bias}
\end{equation}

Now round the raw output with $\mathcal R_\epsilon$.  The decoder-absorption
construction gives the feasible POVM
\[
N_z^{(p)}
:=\sum_{\substack{y:\mathcal R_\epsilon(\vartheta(\rho_y))=z}}
\overline D_y^{(p)},\qquad z\in\mathcal Z_\epsilon.
\]
Eq.~\eqref{eq:grid-rounding-error} changes each signed coordinate payoff
by at most $\epsilon/4$, so Eq.~\eqref{eq:mixed-signed-bias} implies
\[
\sum_{x,j,s}\pi(x,j,s)f_{x,j,s}((N_z^{(p)})_z)
\le\frac{3\epsilon}{4}.
\]
This holds for every mixed adversary $\pi$.  Lemma~\ref{lem:finite-minimax}
therefore gives one POVM $(N_z^{\mathcal F})_z$ such that
\begin{equation}
\max_{x,j}\left|
\theta_j(\rho_x)-\sum_{z\in\mathcal Z_\epsilon}z_j
\operatorname{Tr}(N_z^{\mathcal F}\rho_x^{\otimes N})
\right|\le\frac{3\epsilon}{4}.
\label{eq:finite-state-bias}
\end{equation}
The absolute value follows by maximizing over both signs.  This minimizer
is selected in the full space $\mathcal L_N(\mathcal Z_\epsilon)$. The
prior-dependent recovery measurement supplied one feasible response to
each $\pi$.

We next use the trace-distance tools to choose $\mathcal F$.  For any fixed
POVM $(N_z)_z$, put
\[
B_j:=\sum_{z\in\mathcal Z_\epsilon}z_jN_z.
\]
Since $0\le z_j\le1$, both
$B_j=\sum_z z_jN_z$ and $I-B_j=\sum_z(1-z_j)N_z$ are positive.
Thus $0\preceq B_j\preceq I$, and
$\mathbb E_\rho Z_j=\operatorname{Tr}(B_j\rho^{\otimes N})$.
Lemma~\ref{lem:trace-distance-stability}, applied to $E_j$ and to $B_j$ on
the tensor-product space, gives
\begin{align}
&\left|
[\theta_j(\rho)-\mathbb E_\rho Z_j]
-[\theta_j(\sigma)-\mathbb E_\sigma Z_j]
\right|\notag\\
&\qquad\le
\left|\operatorname{Tr}(E_j(\rho-\sigma))\right|
+\left|\operatorname{Tr}(B_j(\rho^{\otimes N}-\sigma^{\otimes N}))\right|
\notag\\
&\qquad\le(N+1)d_{\mathrm{tr}}(\rho,\sigma).
\label{eq:bias-trace-stability}
\end{align}
Choose a finite trace-distance net of radius $\epsilon/[4(N+1)]$ and use it as
$\mathcal F$ in Eq.~\eqref{eq:finite-state-bias}, for every state $\rho$,
there is a net state $\sigma$ for which Eq.~\eqref{eq:bias-trace-stability}
adds at most $\epsilon/4$ to that bound.  This proves Eq.~\eqref{eq:uniform-bias}
for the same POVM on all states, where the net and the grid are used in the
minimax and continuity arguments, with no probability union bound over
either of their sizes.
\end{proof}

Lemma~\ref{lem:uniform-bias} shows that, under the recovery-measurement approach, the expectation of the estimator is uniformly close to the true expectation value. To turn this small-bias guarantee into a high-probability estimation guarantee, we repeat the measurement independently on multiple disjoint groups of fresh copies and take the empirical mean of the estimates. By Hoeffding's inequality, this empirical mean concentrates around its expectation with high probability. Combining the uniform bias bound with this concentration bound yields the desired shadow-tomography guarantee.

\begin{corollary}
\label{cor:independent-averaging}
For every known list of $M$ effects, every $\epsilon\in(0,1)$, and
$\delta\in(0,1)$, there is a finite-outcome shadow-tomography strategy with
accuracy $\epsilon$ and failure probability $\delta$ using at most
\begin{equation*}
O(\epsilon^{-4}\log(2M)\log(2M/\delta))
\label{eq:intermediate-copy-bound}
\end{equation*}
copies.  The bound is independent of the Hilbert-space dimension.
\end{corollary}

\begin{proof}
Apply Lemma~\ref{lem:uniform-bias} with tolerance $\epsilon/2$. There exists a measurement protocol using
$O(\log(M)/\epsilon^2)$ copies and outputting an estimator $Z_j$ for each observable $E_j$, such that
\begin{align*}
    \sup_{\rho\in\mathcal D(\mathcal H)}
    \max_{1\le j\le M}
    \left|\theta_j(\rho)-\mathbb E_\rho Z_j\right|
    \le \frac{\epsilon}{2}.
\end{align*}
Now repeat this protocol independently on $r$ disjoint groups of fresh copies, and denote the estimator for the $j$-th observable in the $\ell$-th round by $Z_j^{(\ell)}$, where
\[
    r=\left\lceil\frac{2\log(2M/\delta)}{\epsilon^2}\right\rceil.
\]
Define the final estimator by
\[
    \widehat\theta_j
    :=\frac1r\sum_{\ell=1}^r Z_j^{(\ell)}.
\]

For every input state $\rho$ and every $j$, the random variables
$Z_j^{(1)},\ldots,Z_j^{(r)}$ are independent, have the same mean
$\mathbb E_\rho Z_j$, and take values in $[0,1]$.
Therefore, Hoeffding's inequality~\cite[Theorem~2]{hoeffding1963} gives
\begin{equation*}
    \mathbb P_\rho\!\left[
        \left|
        \frac1r\sum_{\ell=1}^r Z_j^{(\ell)}
        -\mathbb E_\rho Z_j
        \right|
        >\frac{\epsilon}{2}
    \right]
    \le
    2\exp(-r\epsilon^2/2)
    \le
    \frac{\delta}{M}.
\end{equation*}
By a union bound over $j\in[M]$,
\[
    \mathbb P_\rho\!\left[
        \max_{1\le j\le M}
        \left|
        \frac1r\sum_{\ell=1}^r Z_j^{(\ell)}
        -\mathbb E_\rho Z_j
        \right|
        >\frac{\epsilon}{2}
    \right]
    \le
    2M\exp(-r\epsilon^2/2)
    \le
    \delta.
\]
Combining this concentration bound with Eq.~\eqref{eq:uniform-bias}, we obtain, for every
$\rho\in\mathcal D(\mathcal H)$,
\[
    \mathbb P_\rho\!\left[
        \max_{1\le j\le M}
        \left|
        \widehat\theta_j-\theta_j(\rho)
        \right|
        >\epsilon
    \right]
    \le\delta.
\]
\end{proof}

\subsection{Geometric precision refinement}
\label{subsec:geometric-refinement}
We now improve the dependence on the target accuracy from
$\varepsilon^{-4}$ to $\varepsilon^{-2}$ by applying
Corollary~\ref{cor:independent-averaging} repeatedly at constant
accuracy. Rather than using the corollary to estimate the original
expectation values directly to accuracy $\varepsilon$, we apply it to
a sequence of observables chosen according to our current estimates and improve the accuracy sequentially.

Suppose that at stage $s$, we have an estimator $c_{j,s}$ with current error bound $\epsilon_s$. For each observable $j$, we construct a block observable $F_{j,s}$ whose energy provides additional information that allows us to refine $c_{j,s}$ and improve its accuracy to the next target level $\epsilon_{s+1}$.

Each invocation of Corollary~\ref{cor:independent-averaging} estimates
the entire list of block observables simultaneously. Its
dimension-independent guarantee allows us to work on
$\mathcal H^{\otimes k_s}$ without paying an additional dimension
factor. The error bounds decrease geometrically, and the total copy count is controlled by
the last few stages. This yields the improvement on $\varepsilon$-dependency.

\begin{theorem}
\label{thm:geometric-refinement}
For every finite-dimensional $\mathcal H$, every set of observables
$0\preceq E_1,\ldots,E_M\preceq I$, and every
$0<\varepsilon\le1$, $0<\delta<1$, there is a finite-outcome collective
POVM and decoder with
\[
\sup_{\rho\in\mathcal D(\mathcal H)}
\mathbb P_\rho\!\left[
\max_{1\le j\le M}|\widehat\theta_j-\theta_j(\rho)|>\varepsilon
\right]\le\delta
\]
using
\begin{equation}
T=O\!\left(\varepsilon^{-2}\log(M)\log(M/\delta)\right)
\label{eq:quadratic-log-copy-bound}
\end{equation}
copies.  The implied constant is independent of $\dim\mathcal H$.
\end{theorem}
\begin{proof}
Consider the following algorithm:
\begin{algorithm}[H]
\caption{Geometric precision refinement}
\label{alg:geometric-refinement}
\begin{algorithmic}[1]

\Require Observables
    $0\preceq E_1,\ldots,E_M\preceq I$ on $\mathcal H$;
    $0<\varepsilon\leq1$, $0<\delta<1$;
    fresh copies of an unknown $\rho\in\mathcal D(\mathcal H)$.

\Ensure Estimates $(\widehat\theta_j)_{j=1}^M$.

\State $c_{j,0}\gets 1/2$ for every $j=1,\ldots,M$

\If{$\varepsilon\geq 1/2$}
    \State \Return $(c_{j,0})_{j=1}^M$
        \Comment{No copies are needed}
\EndIf

\State
 $\epsilon_s\gets\tfrac12\left(\frac{3}{4}\right)^s$

\State $S\gets\min\{s\geq0:\epsilon_s\leq\varepsilon\}$

\For{$s=0,\ldots,S-1$}

    \State $\delta_s\gets(\delta/4)\left(\frac{3}{4}\right)^{S-1-s}$

    \State
        $k_s\gets\lceil C/\epsilon_s^2\rceil$,
        $q_s\gets
        \lceil C\log(2M)\log(2M/\delta_s)\rceil$ with C in Corollary~\ref{cor:independent-averaging} and assume $e^{-C/8}
    \le \frac18$

    \For{$j=1,\ldots,M$}

        \State Define the following observable on
            $\mathcal H^{\otimes k_s}$:

        \Statex
        \hspace{\algorithmicindent}
        \hspace{\algorithmicindent}
        $\displaystyle
        F_{j,s}:=
        \sum_{\substack{
            v\in\{0,1\}^{k_s}\\
            k_s^{-1}\sum_{i=1}^{k_s}v_i\geq c_{j,s}
        }}
        \bigotimes_{i=1}^{k_s}
        \begin{cases}
            E_j,   & v_i=1,\\
            I-E_j, & v_i=0.
        \end{cases}$

    \EndFor

    \State Take $k_sq_s$ fresh copies of $\rho$ and form
        $q_s$ disjoint blocks of $k_s$ copies each.

    \State
        $(\widehat p_{j,s})_{j=1}^M\gets
        \text{Corollary~\ref{cor:independent-averaging} with observables $\{F_{j,s}\}$, accuracy $1/4$, failure probability $\delta_s$}$

    \For{$j=1,\ldots,M$}

        \If{$\widehat p_{j,s}\geq 1/2$}
            \State
                $c_{j,s+1}\gets
                (c_{j,s}+\epsilon_s/2)$
        \Else
            \State
                $c_{j,s+1}\gets
                (c_{j,s}-\epsilon_s/2)$

        \EndIf

    \EndFor

\EndFor

\State \Return
    $(c_{j,S})_{j=1}^{M}$

\end{algorithmic}
\end{algorithm}

\noindent\textit{Correctness:} 
Fix $\rho$ and write $\theta_j=\theta_j(\rho)$. Let
\[
p_{j,s}
:=\Tr\!\left[\rho^{\otimes k_s}F_{j,s}\right]
\]
denote the expectation of the observable $F_{j,s}$. By the product structure of $F_{j,s}$, we have
\[
p_{j,s}
=
\Pr\!\left[
\frac{1}{k_s}\sum_{i=1}^{k_s}V_i
\ge c_{j,s}
\right],
\qquad
V_1,\ldots,V_{k_s}
\stackrel{\mathrm{i.i.d.}}{\sim}
\operatorname{Bernoulli}(\theta_j).
\]
Since $V_1,\ldots,V_{k_s}$ are independent Bernoulli random variables with mean $\theta_j$, Hoeffding's inequality implies that, for every $t>0$,
\[
\Pr\!\left[
\frac{1}{k_s}\sum_{i=1}^{k_s}V_i-\theta_j\ge t
\right]
\le e^{-2k_st^2},
\qquad
\Pr\!\left[
\theta_j-\frac{1}{k_s}\sum_{i=1}^{k_s}V_i\ge t
\right]
\le e^{-2k_st^2}.
\]
Consequently, let $\gamma_s=\epsilon_s/4$ and
$k_s=\lceil C/\epsilon_s^2\rceil$, we have:
\begin{equation}
\begin{aligned}
    \theta_j \le c_{j,s}-\gamma_s
    &\quad\Longrightarrow\quad
    p_{j,s}
    \le
    \Pr\!\left(
        \frac{1}{k_s}\sum_{i=1}^{k_s}V_i-\theta_j
        \ge \gamma_s
    \right)
    \le e^{-2k_s\gamma_s^2}
    \le \frac18,\\
    \theta_j \ge c_{j,s}+\gamma_s
    &\quad\Longrightarrow\quad
    1-p_{j,s}
    \le
    \Pr\!\left(
        \theta_j-\frac{1}{k_s}\sum_{i=1}^{k_s}V_i
        \ge \gamma_s
    \right)
    \le e^{-2k_s\gamma_s^2}
    \le \frac18.
\end{aligned}
\label{eq:threshold-tail-separation}
\end{equation}
where
\[
e^{-2k_s\gamma_s^2}
=
e^{-k_s\epsilon_s^2/8}
\le
e^{-C/8}
\le
\frac18,
\]
Apply Corollary~\ref{cor:independent-averaging} on the Hilbert space
$\mathcal H^{\otimes k_s}$ to this list, at accuracy $1/4$ and failure
probability $\delta_s$, which requires
\begin{equation}
q_s:=\left\lceil C\log(2M)\log(2M/\delta_s)\right\rceil
\label{eq:refinement-block-budget}
\end{equation}
fresh copies of the block state $\rho^{\otimes k_s}$. Then we
update the estimator by
\begin{equation*}
c_{j,s+1}:=
\begin{cases}
c_{j,s}+\epsilon_s/2,&\widehat p_{j,s}\ge1/2,\\
c_{j,s}-\epsilon_s/2,&\widehat p_{j,s}<1/2.
\end{cases}
\end{equation*}

At each stage $s$, we say that stage $s$ is \emph{successful} if the protocol call of
Corollary~\ref{cor:independent-averaging} succeeds, namely,
\[
    \max_j |\widehat p_{j,s}-p_{j,s}|\le \frac14.
\]
We prove by induction that, if all stages are successful, then the
current estimator satisfies
\[
    |c_{j,s}-\theta_j|\le \epsilon_s
\]
For the base case $s=0$, recall that $c_{j,0}=1/2$ and
$\epsilon_0=1/2$. Since $\theta_j\in[0,1]$, we have
\[
    |c_{j,0}-\theta_j|\le \epsilon_0.
\]
At stage $s$, suppose that all previous stages and stage $s$
are successful.
By the induction hypothesis and the success of stage $s$, we have:
\[
    |c_{j,s}-\theta_j|\le \epsilon_s,\qquad |\widehat p_{j,s}-p_{j,s}|\le \frac14
\]
First, if $\widehat p_{j,s}\ge 1/2$. Then
\[
    p_{j,s}
    \ge \widehat p_{j,s}
        -|\widehat p_{j,s}-p_{j,s}|
    \ge \frac12-\frac14
    =\frac14.
\]
Since $p_{j,s}>1/8$, by
Eq.~(\ref{eq:threshold-tail-separation}), we have
\[
    \theta_j>c_{j,s}-\frac{\epsilon_s}{4}.
\]
Combining this with the induction hypothesis yields
\[
    c_{j,s}-\frac{\epsilon_s}{4}
    <\theta_j
    \le c_{j,s}+\epsilon_s.
\]
In this case, the algorithm sets $c_{j,s+1}=c_{j,s}+\frac{\epsilon_s}{2}$. Therefore,
\[
    -\frac{3\epsilon_s}{4}
    <\theta_j-c_{j,s+1}
    \le \frac{\epsilon_s}{2} \leq \frac{3\epsilon_s}{4}
\]
Similarly, suppose that $\widehat p_{j,s}<1/2$. We also have:
\[
   -\frac{3\epsilon_s}{4}
    <\theta_j-c_{j,s+1} \le \frac{3\epsilon_s}{4}
\]
Thus, in either case,
\[
    |c_{j,s+1}-\theta_j|
    <\frac{3\epsilon_s}{4}
    =\epsilon_{s+1},
\]
which completes the induction.

Hence, suppose that all $S$ stages are successful. By the previous induction argument, we have
\[
    |c_{j,S}-\theta_j|
    < \epsilon_S
    \le \varepsilon,
\]
which gives the desired estimation accuracy.

It remains to bound the overall success probability. By a union bound, the probability that all $S$ stages succeed is at least
\begin{align*}
    1-\sum_{s=0}^{S-1}\delta_s
    &=
    1-\frac{\delta}{4}
    \sum_{s=0}^{S-1}
    \left(\frac{3}{4}\right)^s \\
    &=
    1-\delta
    \left(
        1-\left(\frac{3}{4}\right)^S
    \right)
    > 1-\delta.
\end{align*}
Therefore, with probability at least $1-\delta$, the final estimates satisfy $
    |c_{j,S}-\theta_j|<\varepsilon$ simultaneously for all $j$.

\noindent\textit{Sample Complexity:} The stage $s$ uses $k_sq_s$ copies of $\rho$. The total sample complexity is
\begin{equation}
T=\sum_{s=0}^{S-1}k_sq_s
=\sum_{s=0}^{S-1}
\left\lceil\frac{C}{\epsilon_s^2}\right\rceil
\left\lceil C\log(2M)\log(2M/\delta_s)\right\rceil.
\label{eq:refinement-total-budget}
\end{equation}
Write $\ell=S-1-s$. Since $\epsilon_{S-1}>\varepsilon$, we have:
\begin{equation*}
\epsilon_s^{-2}=\epsilon_{S-1}^{-2}\left(\frac{3}{4}\right)^{2\ell}
<\varepsilon^{-2}\left(\frac{3}{4}\right)^{2\ell},
\qquad
\log(2M/\delta_s)=\log(8M/\delta)+\ell\log(4/3).
\end{equation*}
Summing the two geometric series gives
\begin{equation}
T\le C\varepsilon^{-2}\log(2M)
\left(\log(8M/\delta)
+3\right).
\label{eq:refinement-summed-cost}
\end{equation}
for some constant $C$, which proves Eq.~\eqref{eq:quadratic-log-copy-bound}. 
\end{proof}

Theorem~\ref{thm:geometric-refinement} proves Theorem~\ref{thm:main}.

%% file: sections/binomial_distribution.tex
\section{Moment bound for binomial distribution}

Define the numerical constants
\[
  b_2\coloneqq\frac14,
  \qquad
  b_q\coloneqq
  2q\int_0^\infty t^{q-1}\mathrm e^{-2t^2}\dd t
  \quad(q>2).
\]
\begin{lemma}
  \label{lem:binomial-moments}
  For every integer $q\geq2$ and $n\geq1$,
  \begin{equation*}
    \sup_{x\in[0,1]}
    \mathbb{E}_x\left|\frac{K}{n}-x\right|^q
    \leq b_qn^{-q/2},
    \qquad
    K\sim\operatorname{Bin}(n,x).
      \end{equation*}
\end{lemma}

\begin{proof}
  For $q=2$,
  \[
    \mathbb E_x|K/n-x|^2
    =\frac{x(1-x)}n
    \leq\frac1{4n}.
  \]
  For $q>2$, Hoeffding's inequality~\cite[Theorem~2]{hoeffding1963} gives
  \[
    \mathbb P_x\left[|K/n-x|\geq t\right]
    \leq2\mathrm e^{-2nt^2}.
  \]
  Integrating the tail,
  \[
    \begin{aligned}
    \mathbb E_x|K/n-x|^q
    &=q\int_0^\infty
      t^{q-1}\mathbb P_x[|K/n-x|\geq t]\dd t\\
    &\leq
    2q\int_0^\infty t^{q-1}\mathrm e^{-2nt^2}\dd t
    =b_qn^{-q/2}.
    \end{aligned}
  \]
\end{proof}

Also, we could have an upper bound on the coefficients $b_q$.
\begin{lemma}
  For every integer \(q\geq2\),
  \[
    b_q\leq\left(\frac q2\right)^{q/2}.
  \]
  Consequently, for every \(n\geq1\),
  \[
    \sup_{x\in[0,1]}
    \mathbb{E}_x\left|
      \frac{K}{n}-x
    \right|^q
    \leq
    \left(\frac{q}{2n}\right)^{q/2},
    \qquad
    K\sim\operatorname{Bin}(n,x).
  \]
  \label{lem: binomial-coefficient-bound}
\end{lemma}
\begin{proof}
   For \(q=2\), the claim follows from
  \[
    b_2=\frac14\leq1.
  \]
  Now suppose that \(q>2\). The change of variables \(u=2t^2\) gives
  \[
    t=\left(\frac u2\right)^{1/2},
    \qquad
    \dd t
    =
    2^{-3/2}u^{-1/2}\dd u,
  \]
  Hence
  \begin{align*}
    b_q
    &=2q\int_0^\infty t^{q-1}\mathrm e^{-2t^2}\dd t\\
    &=2q\int_0^\infty
    2^{-q/2-1}
    u^{q/2-1}\mathrm e^{-u}\dd u\\
    &=
    q2^{-q/2}
    \int_0^\infty
    u^{q/2-1}\mathrm e^{-u}\dd u.\\
    &= q2^{-q/2}\Gamma\left(\frac q2\right)\\
    &= 2^{1-q/2}
    \Gamma\left(\frac q2+1\right)
  \end{align*}
  where $\Gamma(x)$ is the Gamma function
  \[
    \Gamma(s)
    \coloneqq
    \int_0^\infty
    u^{s-1}\mathrm e^{-u}\dd u,
    \qquad s>0.
  \]
  Using the standard inequality
  \[
    \Gamma(s+1)\leq s^s,
    \qquad s\geq1,
  \]
  with \(s=q/2\), we conclude that
  \[
    \begin{aligned}
    b_q
    &\leq
    2^{1-q/2}
    \left(\frac q2\right)^{q/2}\\
    &\leq
    \left(\frac q2\right)^{q/2},
    \end{aligned}
  \]
  since \(2^{1-q/2}\leq1\) for \(q\geq2\).

  Finally, Lemma~\ref{lem:binomial-moments} yields
  \[
    \begin{aligned}
    \sup_{x\in[0,1]}
    \mathbb{E}_x\left|
      \frac{K}{n}-x
    \right|^q
    &\leq
    b_qn^{-q/2}\\
    &\leq
    \left(\frac q2\right)^{q/2}n^{-q/2}\\
    &=
    \left(\frac{q}{2n}\right)^{q/2}.
    \end{aligned}
  \]
\end{proof}

%% file: sections/recovery_appendix.tex
\section{Support, entropy, and recovery calculations}
\label{app:recovery-technical}

All spaces below are finite-dimensional.  Inverse powers are taken on
supports, and complex powers with positive real part are zero on kernels.

\subsection{Support, trace distance, and root fidelity}

\begin{lemma}
\label{lem:support-distance-fidelity}
For every bipartite state $\omega_{UV}$,
\begin{equation}
\supp\omega_{UV}\subseteq\supp\omega_U\otimes\supp\omega_V.
\label{eq:marginal-support}
\end{equation}
For density matrices $\rho,\sigma$ on the same space,
\begin{equation}
d_{\rm tr}(\rho,\sigma)
=\max_{0\preceq E\preceq I}
\left|\operatorname{Tr}(E(\rho-\sigma))\right|.
\label{eq:trace-distance-effects}
\end{equation}
Trace distance contracts under quantum channels, and
\begin{equation}
d_{\rm tr}(\rho,\sigma)\le\sqrt{1-F(\rho,\sigma)^2}.
\label{eq:root-fidelity-distance}
\end{equation}
\end{lemma}

\begin{proof}
If $Q_U$ projects onto $\ker\omega_U$, then
\[
0=\operatorname{Tr}((Q_U\otimes I)\omega_{UV})
=\|(Q_U\otimes I)\omega_{UV}^{1/2}\|_2^2.
\]
Thus $(Q_U\otimes I)\omega_{UV}=0$.  Applying the same argument on
$V$ proves Eq.~\eqref{eq:marginal-support}.
The Jordan-decomposition argument in
Lemma~\ref{lem:trace-distance-stability} bounds every effect expectation
by $\tfrac12\|\rho-\sigma\|_1$, with equality for the projection onto
the positive eigenspace of $\rho-\sigma$.  This proves
Eq.~\eqref{eq:trace-distance-effects}.  The adjoint of a quantum channel
maps effects to effects, so that identity also proves contraction.
Equation~\eqref{eq:root-fidelity-distance} is
\cite[Theorem~1, Eq.~(46)]{fuchsvandegraaf1999} in the root-fidelity
convention, including singular states.  We also use
\[
F(\rho,\sigma)=\|\sqrt\rho\sqrt\sigma\|_1
\le\|\sqrt\rho\|_2\|\sqrt\sigma\|_2=1.
\]
\end{proof}

\subsection{The recovery map on a classical register}
\label{app:recovery-specialization}

Let $\omega_{JKC}$ be a state with $J$ classical, and write
\[
\omega_{JK}=\sum_a\ketbra{a}{a}\otimes C_a,
\qquad B=\sum_aC_a=\omega_K,
\qquad q_a=\operatorname{Tr}C_a,
\qquad P=\Pi_{\supp B}.
\]
Each $C_a$ is positive and supported on $P$, since $C_a\preceq B$.
For $u\in\mathbb R$, put $\alpha=(1-iu)/2$, so $\bar\alpha=(1+iu)/2$.
On $\supp B$ the rotated
map in Theorem~\ref{thm:universal-recovery} is
\[
\mathcal R_u(X)=\sum_a\ketbra{a}{a}\otimes
C_a^\alpha B^{-\alpha}X B^{-\bar\alpha}C_a^{\bar\alpha}.
\]
Every block is a positive congruence, so the map is completely positive.
Since $\alpha+\bar\alpha=1$, cyclicity gives
\[
\operatorname{Tr}\mathcal R_u(X)
=\operatorname{Tr}\!\left(
B^{-\bar\alpha}\Bigl(\sum_aC_a\Bigr)B^{-\alpha}X\right)
=\operatorname{Tr}(PX)=\operatorname{Tr}X.
\]
Therefore $\mathcal R=\int_{\mathbb R}\beta_0(u)\mathcal R_u\,du$
preserves trace on this support.  Its full-space completion is
\begin{equation}
\widehat{\mathcal R}(X)
=\mathcal R(PXP)+\operatorname{Tr}((I-P)X)\omega_{JK}.
\label{eq:classical-recovery-completion}
\end{equation}
Both terms are completely positive, and their traces sum to
$\operatorname{Tr}(PX)+\operatorname{Tr}((I-P)X)=\operatorname{Tr}X$.
Thus this completion is a channel.

For the reference $\tau=\omega_{JK}\otimes\omega_C$ and channel
$\mathcal N=\operatorname{Tr}_J$, Eq.~\eqref{eq:marginal-support}
gives $\supp\omega_{JKC}\subseteq\supp\tau$ and
$\supp\omega_{KC}\subseteq\supp B\otimes\supp\omega_C$.
Set $Q=\Pi_{\supp\omega_C}$.  The left product in the $a$-block of
Eq.~\eqref{eq:rotated-recovery-map} is
\[
(C_a^\alpha\otimes\omega_C^\alpha)
(B^{-\alpha}\otimes\omega_C^{-\alpha})
=C_a^\alpha B^{-\alpha}\otimes Q.
\]
Its right product is the adjoint.  Hence, on every operator supported
on $P\otimes Q$, the rotated map equals
$\mathcal R_u\otimes\operatorname{id}_C$, including correlated
operators.  This proves the spectator identity for singular
$\omega_C$ as well.  It agrees with
\cite[Remark~2.4(4)]{junge2018recovery} after restricting the spectator
to its support, where its reference state is faithful.
The complement term in Eq.~\eqref{eq:classical-recovery-completion}
vanishes on the actual input $\omega_{KC}$, so
$(\widehat{\mathcal R}\otimes\operatorname{id}_C)(\omega_{KC})$
is the recovered state used in Theorem~\ref{thm:universal-recovery}.

Tracing out the recovered $K$ combines the powers of $C_a$.
The full-space effect for label $a$ is
\[
\int_{\mathbb R}\beta_0(u)
B^{-(1+iu)/2}C_aB^{-(1-iu)/2}\,du+q_a(I-P).
\]
It is linear in the unnormalized block $C_a$, including its trace
$q_a$.  A zero $q_a$ gives $C_a=0$ and a zero effect.  Grouping blocks
by a stochastic matrix therefore groups these effects by that same
matrix, proving Eq.~\eqref{eq:coarse-recovery-consistency}.

For a general channel in Theorem~\ref{thm:universal-recovery}, the same
trace-adjoint calculation gives, on $\supp B$ with $B=\mathcal N(\tau)$,
\[
\operatorname{Tr}\mathcal R^{[u]}_{\tau,\mathcal N}(X)
=\operatorname{Tr}\!\left(
B^{-(1+iu)/2}B B^{-(1-iu)/2}X\right)=\operatorname{Tr}X.
\]
Compression to this support and sending its complement to a fixed
state gives a channel completion as in
Eq.~\eqref{eq:classical-recovery-completion}.  Zero extension alone
need not preserve trace on the full space.

\subsection{The entropy of a classical label}

We use the entropy notation from Section~\ref{sec:recovery-analysis},
suppressing the state subscript when it is fixed.

\begin{lemma}
\label{lem:classical-entropy-budget}
If $J$ is classical with at most $k$ possible values and
$Q_1,\ldots,Q_N$ are finite-dimensional quantum registers, then
\begin{equation}
\sum_{t=0}^{N-1}I(J:Q_{t+1}\mid Q_1\cdots Q_t)
=I(J:Q_1\cdots Q_N)\le H(J)\le\log k.
\label{eq:classical-entropy-chain}
\end{equation}
Every summand is nonnegative.
\end{lemma}

\begin{proof}
The identity $I(J:C\mid K)=I(J:KC)-I(J:K)$ telescopes to the first
part of Eq.~\eqref{eq:classical-entropy-chain}.
The nonnegativity of each summand follows from strong
subadditivity~\cite[Theorem~2(i)]{liebruskai1973}.

For the upper bound, combine the quantum registers into $Q$ and write
\[
\omega_{JQ}=\sum_aq_a\ketbra{a}{a}\otimes\sigma_a,
\qquad B=\sum_aq_a\sigma_a,
\]
omitting zero probabilities.  The classical-block eigenvalues give
$S(\omega_{JQ})=H(q)+\sum_aq_aS(\sigma_a)$, and therefore
\begin{equation}
I(J:Q)=S(B)-\sum_aq_aS(\sigma_a)
=\sum_aq_aD(\sigma_a\Vert B).
\label{eq:classical-holevo-expansion}
\end{equation}
For $C\succ0$, diagonalization gives
\begin{equation}
\log C=\int_0^\infty\left((1+v)^{-1}I-(C+vI)^{-1}\right)\,dv.
\label{eq:log-resolvent}
\end{equation}
If $C\succeq D\succ0$, conjugating by $(D+vI)^{-1/2}$ gives
$(C+vI)^{-1}\preceq(D+vI)^{-1}$, so
Eq.~\eqref{eq:log-resolvent} implies $\log C\succeq\log D$.
Apply this to $B+wI\succeq q_a\sigma_a+wI$ on $\supp B$, take the
trace against $\sigma_a$, and let $w\downarrow0$.  Only positive
eigenvalues of $\sigma_a$ contribute, giving
$D(\sigma_a\Vert B)\le-\log q_a$.
Equation~\eqref{eq:classical-holevo-expansion} now gives $I(J:Q)\le H(q)$.
Finally, concavity of the scalar logarithm gives
\[
H(q)=\sum_aq_a\log(1/q_a)
\le\log\left(\sum_{a:q_a>0}q_a/q_a\right)\le\log k.
\]
This proves the bound also for singular quantum marginals.
\end{proof}

\section{A scalar quadratic comparison on the same ensemble}
\label{app:recovery-quadratic}

We record a comparison of the two label measurements on one fixed
input space.  It uses the optimal scalar-estimation formula already
proved in Proposition~\ref{prop:personick}.

\begin{lemma}
\label{lem:recovery-scalar-kernel}
For $a,b>0$, put
\[
\Lambda(a,b):=
\begin{cases}
\frac{\log a-\log b}{a-b},&a\ne b,\\
\frac1a,&a=b.
\end{cases}
\]
If $B=\sum_ad_a\ketbra{a}{a}\succ0$, then
\begin{equation}
\left[\int_{\mathbb R}\beta_0(u)
B^{-(1+iu)/2}AB^{-(1-iu)/2}\,du\right]_{ab}
=\Lambda(d_a,d_b)A_{ab},
\label{eq:recovery-spectral-kernel}
\end{equation}
and
\begin{equation}
\frac2{a+b}\le\Lambda(a,b)\le\frac1{\sqrt{ab}}.
\label{eq:recovery-kernel-comparison}
\end{equation}
\end{lemma}

\begin{proof}
The Fourier transform of $\beta_0$ is
\begin{equation}
\int_{\mathbb R}\beta_0(u)e^{ivu}\,du=\frac{v}{\sinh v},
\label{eq:recovery-density-fourier}
\end{equation}
with value one at $v=0$.  To verify it for $v>0$, integrate
$f(z)=e^{ivz}/\cosh^2(\pi z/2)$ around the rectangle with vertices
$-R,R,R+2i,-R+2i$ around the pole, then $f(z+2i)=e^{-2v}f(z)$, and vertical integral is zero as $R\rightarrow \infty$.  Its only pole inside
the rectangle is the double pole at $z=i$, so, with $z = i + w$, taking the limit in the residue theorem gives
\[
\cosh^{-2}(\pi(i+w)/2)=-\frac4{\pi^2w^2}+O(1),
\]
so the residue of $f$ at $i$ is $-4iv e^{-v}/\pi^2$.  The residue
theorem yields
\[
(1-e^{-2v})\int_{\mathbb R}\frac{e^{ivu}}{\cosh^2(\pi u/2)}\,du
=\frac{8ve^{-v}}\pi.
\]
Multiplying by $\pi/4$ proves Eq.~\eqref{eq:recovery-density-fourier}, evenness gives the case $v<0$, and dominated convergence gives $v=0$.
Taking $v=(\log d_b-\log d_a)/2$, the coefficient multiplying
$A_{ab}$ in the LHS of Eq.~\eqref{eq:recovery-spectral-kernel}
is
\[
\frac1{\sqrt{d_ad_b}}\int_{\mathbb R}\beta_0(u)e^{ivu}\,du
=\frac1{\sqrt{d_ad_b}}\frac{v}{\sinh v}
=\Lambda(d_a,d_b).
\]
For Eq.~\eqref{eq:recovery-kernel-comparison}, the integral identity
\[
\Lambda(a,b)=\int_0^1\frac{dt}{ta+(1-t)b}
\]
and convexity of $v\mapsto1/v$ give the lower bound $2/(a+b)$.
The upper bound follows from $v/\sinh v\le1$, including its value
at zero.  These arguments include the diagonal case $a=b$ by continuity.
\end{proof}

\begin{proposition}
\label{prop:recovery-quadratic}
Fix the same ensemble $\{(p_x,\rho_x)\}$ on the same
finite-dimensional input space, and give it real scalar labels $z_x$.
Let $R_{\rm PGM}$ and $R_{\rm rec}$ be the mean-square risks when
the ordinary PGM and the averaged recovery measurement, respectively,
report $z_y$ on outcome $y$.  Let $R_{\rm opt}$ be the infimum over
all scalar strategies on that input space.  Then
\begin{equation}
R_{\rm PGM}\le R_{\rm rec}\le2R_{\rm opt}.
\label{eq:recovery-quadratic-comparison}
\end{equation}
\end{proposition}

\begin{proof}
Set
\[
B:=\sum_xp_x\rho_x,\qquad
A:=\sum_xp_xz_x\rho_x,\qquad
\]
All calculations are on $\supp B$.  The effects are
\[
G_y=p_yB^{-1/2}\rho_yB^{-1/2},\qquad
D_y=p_y\int_{\mathbb R}\beta_0(u)
B^{-(1+iu)/2}\rho_yB^{-(1-iu)/2}\,du.
\]
The common completion $p_y(I-\Pi_{\supp B})$ has zero probability
on the positive-weight source states.  The PGM joint distribution
has the symmetry and marginals used in Eq.~\eqref{eq:pgm-risk-formula}.
Lemma~\ref{lem:recovery-normalization} gives the same facts for the
recovery measurement.  Therefore the two marginal square terms in
each quadratic loss sum to $2\sum_xp_xz_x^2$.  In an eigenbasis
$B=\sum_ad_a\ketbra{a}{a}$ on its support, Eq.~\eqref{eq:pgm-risk-formula}
and Lemma~\ref{lem:recovery-scalar-kernel} give
\begin{align*}
R_{\rm PGM}
&=2\sum_xp_xz_x^2-2\operatorname{Tr}(AB^{-1/2}AB^{-1/2})
=2\sum_xp_xz_x^2-2\sum_{a,b}\frac{|A_{ab}|^2}{\sqrt{d_ad_b}},
\\
R_{\rm rec}
&=2\sum_xp_xz_x^2-2\int_{\mathbb R}\beta_0(u)
\operatorname{Tr}\!\left(AB^{-(1+iu)/2}AB^{-(1-iu)/2}\right)\,du
\notag\\
&=2\sum_xp_xz_x^2-2\sum_{a,b}\Lambda(d_a,d_b)|A_{ab}|^2.
\end{align*}
By Proposition~\ref{prop:personick}, the optimal scalar risk is
$\sum_xp_xz_x^2-\operatorname{Tr}(LA)$, where $BL+LB=2A$.
Thus $L_{ab}=2A_{ab}/(d_a+d_b)$ and
\begin{equation*}
R_{\rm opt}=\sum_xp_xz_x^2-\sum_{a,b}\frac{2|A_{ab}|^2}{d_a+d_b}.
\end{equation*}
Applying Eq.~\eqref{eq:recovery-kernel-comparison} to the three
spectral sums proves Eq.~\eqref{eq:recovery-quadratic-comparison}.
\end{proof}

The comparison in Eq.~\eqref{eq:recovery-quadratic-comparison} fixes
the ensemble, the input space, and a scalar quadratic loss.
Algorithm~\ref{alg:sequential-pgm} changes the ensemble according to its posterior history,
whereas Algorithm~\ref{alg:recovery-pgm} uses its original prior on a
random whole prefix.  The scalar inequality consequently gives no
performance ordering for these complete protocols or for their
simultaneous estimation guarantees, which is the main motivation that we keep Algorithm~\ref{alg:sequential-pgm}, since ideally the PGM attains a better quadratic loss, while we did not prove that the sequential PGM can have an $O(\frac{\log^2(M)}{\varepsilon^2})$ sample complexity.